\documentclass[letterpaper,twocolumn,10pt]{article}
\usepackage{zhanggroup}

\usepackage{times}
\usepackage{latexsym}
\usepackage[T1]{fontenc}
\usepackage[utf8]{inputenc}
\usepackage{microtype}
\usepackage{inconsolata}
\usepackage{graphicx}

\usepackage{booktabs}
\usepackage{amsmath}
\usepackage{amssymb}
\usepackage{amsthm}
\usepackage{xspace}
\usepackage[flushleft]{threeparttable}
\usepackage{multirow}
\usepackage{pifont}
\usepackage{wasysym}
\usepackage{subcaption}
\usepackage{cleveref}
\Crefname{claim}{Claim}{Claims}

\theoremstyle{plain}
\newtheorem{theorem}{Theorem}[section]

\newtheorem{lemma}{Lemma}[section]
\newtheorem{claim}{Claim}[section]
\theoremstyle{definition}

\newtheorem{hypothesis}[theorem]{Hypothesis}
\theoremstyle{remark}

\Crefname{hypothesis}{Hypothesis}{Hypotheses}
\crefname{hypothesis}{Hypothesis}{Hypotheses}

\newcommand{\mypara}[1]{\noindent{\bf {#1}.}}

\newcommand{\cmark}{\ding{51}}
\newcommand{\xmark}{\ding{55}}
\newcommand{\soft}{\texttt{PlugAE}\xspace}
\newcommand{\na}{-\xspace}
\newcommand{\gpt}{gpt-4-turbo-2024-04-09\xspace}
\newcommand{\claude}{claude-3-5-sonnet-20240620\xspace}
\newcommand{\token}{copyright token\xspace}
\newcommand{\tokens}{copyright tokens\xspace}
\newcommand{\black}{$\blacksquare$\xspace}
\newcommand{\white}{$\square$\xspace}

\date{}

\title{\bf The Challenge of Identifying the Origin of Black-Box Large\\ Language Models}

\author{
Ziqing Yang\textsuperscript{1}\ \ \
Yixin Wu\textsuperscript{1}\ \ \
Yun Shen\textsuperscript{2}\ \ \
Wei Dai\textsuperscript{3}\ \ \
Michael Backes\textsuperscript{1}\ \ \
Yang Zhang\textsuperscript{1}\thanks{Corresponding author.}
\\
\\
\textsuperscript{1}\textit{CISPA Helmholtz Center for Information Security}\ \ \
\textsuperscript{2}\textit{HPE}\ \ \
\textsuperscript{3}\textit{TikTok Inc.}
}

\begin{document}

\maketitle

\begin{abstract}
The tremendous commercial potential of large language models (LLMs) has heightened concerns over their unauthorized use.
{To address this, we focus on the task of identifying the origin of black-box LLMs.
We further propose \soft, an effective and efficient identification method that proactively leverages LLM-specific adversarial embeddings and allows users to customize \tokens on a targeted query set.}
Extensive experiments demonstrate that \soft outperforms both state-of-the-art model watermarking and fingerprinting methods in accuracy and robustness.
We further analyze its stealthiness and reliability from three complementary perspectives and conduct ablation studies under various configurations, confirming its practicality for real-world misuse detection.\footnote{Our code is available at \url{https://github.com/TrustAIRLab/PlugAE}.}
\end{abstract}

\section{Introduction}

Large language models (LLMs) have become integral parts in many domains~\cite{TLIMLLRGHARJGL23,JSMBCCBLLSLLSSLWLS23,T24,GBWBKTJIMWAAACCDEGHKMMMNNPPRSSSSSWDLRZDLSSH24}.
Their commercial potential has led to increasing concerns over their unauthorized use, raising broader issues of fairness, accountability, and governance.
Identifying the origin of a black-box LLM emerges as an intrinsic solution to the problem.
Specifically, we refer to the LLM under the defender's control as the \emph{candidate model}, and the black-box LLM under investigation as the \emph{suspect model}.
Our goal is to identify whether the {suspect} LLM is derived from the candidate LLM, as shown in \Cref{figure:scenario}.
{
Compared with \emph{only} identifying the candidate LLM without considering its derivatives, our task is more practical and challenging, as the malicious third party can further fine-tune the candidate LLM for specific usage.
}

\begin{figure}[!t]
\centering
\includegraphics[width=\columnwidth]{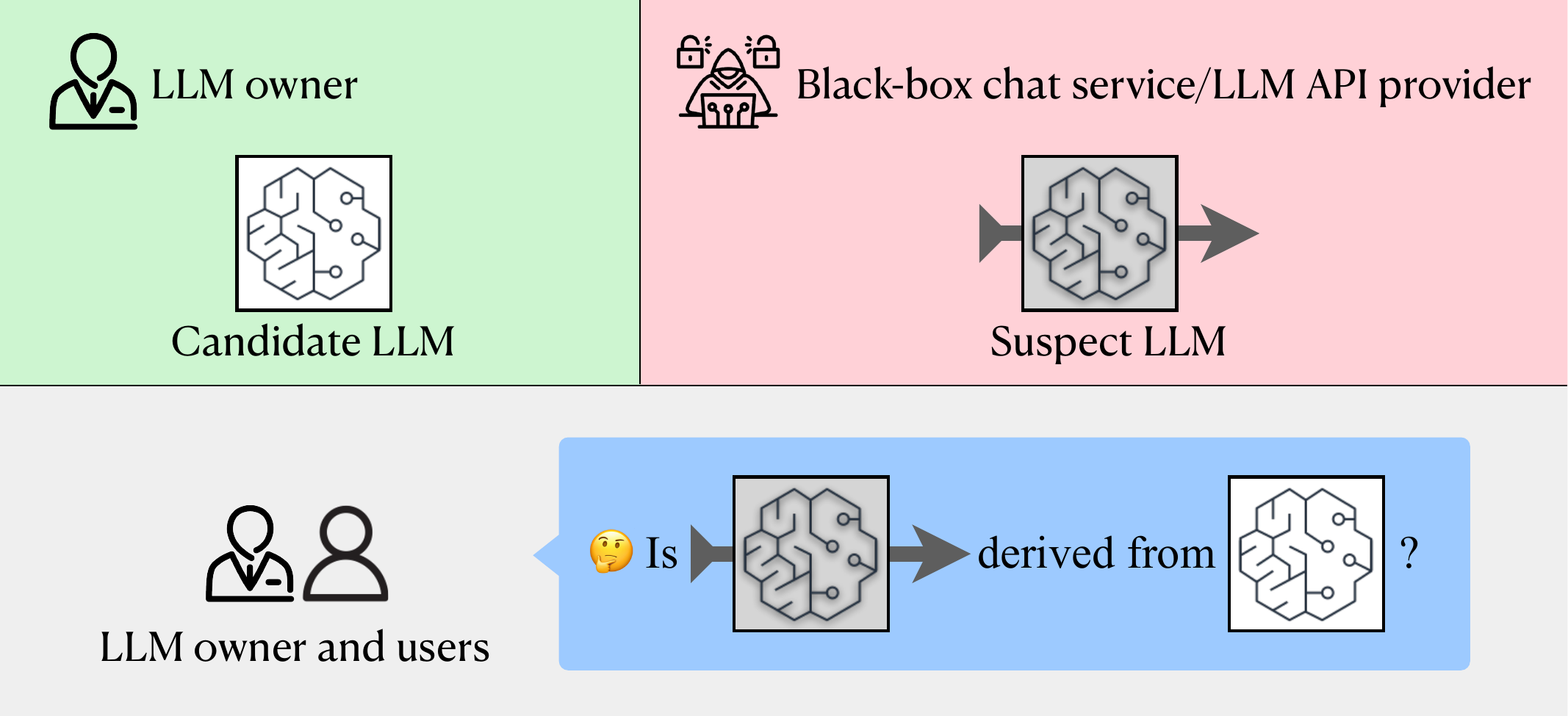}
\caption{Scenario of black-box LLM identification, where the owner aims to verify unauthorized use of their LLM in a black-box chat service.
}
\label{figure:scenario}
\end{figure}

{There have been several related works on LLM identification, which can be roughly divided into model watermarking and fingerprinting.\footnote{The terminology in this area is not always consistent.}
Following~\cite{CJG21, XYWLZZZXKLH25, YRLYYWL26}, we consider the work as model watermarking when they require the model owner to embed signals into the model proactively, e.g., via backdooring~\cite{GHZCH22, LCLDZL23, XWMKXC24}; otherwise, we consider them as fingerprinting~\cite{DCS21, ZZWL23, GULYO24, JZSLH24}, e.g., optimizing adversarial prompts for verification.
However, backdoor-based model watermarking methods inevitably alter the model's parameters and are computationally expensive.
Nevertheless, existing fingerprinting methods fall short in circumstances where they only have black-box access to the suspect model.}

\begin{table*}[!t]
\centering
\setlength{\tabcolsep}{1.1mm}
\scalebox{0.80}{
\begin{tabular}{ll|cccc|c}
\toprule
& & \multicolumn{4}{c|}{Candidate Model} & {Suspect Model}\\
Category & {Method} & Access & Proactive & Weights & Adaptive & Access \\
\midrule
\multirow{2}{*}{Fingerprinting} & Parameter-Based & \white & \xmark & \cmark & \cmark & \white \\
& Query-Based & \white & \xmark & \cmark & \cmark & \black \\
\midrule
{Model Watermarking} & Backdoor-Based & \white & \cmark & \xmark & \xmark & \black \\
\midrule
 & {\soft (Ours)} & \white & \cmark & \cmark & \cmark & \black \\
\bottomrule
\end{tabular}
}
\caption{Different types of identification methods.
``Access'' represents that the method requires either black-box access (\black) or white-box access (\white) to the corresponding LLM.
``Proactive'' denotes whether the LLM will be proactively modified, ``Weights'' represents whether the model's main transformer layers are unaltered, while ``Adaptive'' assesses whether the method, particularly proactive ones, remains applicable to existing derivative models
}
\label{table:method_type}
\end{table*}

To address these challenges, we propose \soft, which proactively \underline{\textbf{Plug}}s the \underline{\textbf{A}}dversarial token \underline{\textbf{E}}mbeddings into the LLM's token-embedding layer for identification, {leaving the main transformer layers unchanged}.
As shown in \Cref{figure:framework}, \soft first optimizes a sequence of adversarial token embeddings on a pre-defined query set $Q$ based on the candidate LLM.
Then, \soft plugs these token embeddings paired with customized \tokens into the model's token-embedding layer.
This process only deals with the related token embeddings and does not alter {the main transformer layers}, thus preserving the model's utility.
Given a suspect LLM, we can identify whether it is a derivative of the candidate LLM by querying it with the query set and predefined \tokens.
We theoretically prove that our \soft is more likely to achieve the optimal than current fingerprinting methods, as they utilize a discrete optimization~\cite{GULYO24, JZSLH24}.
Our method is also highly efficient, e.g., we reduce trainable parameters from 7B to just 4K compared to backdoor-based model watermarking.

To demonstrate the effectiveness of \soft, we experiment on three candidate LLMs and 32 suspect models (30 LLMs and two real-world LLM APIs).
We compare our method with both state-of-the-art model watermarking (Instructional Fingerprint (IF)~\cite{XWMKXC24}) and fingerprinting methods (TRAP~\cite{GULYO24} and ProFlingo~\cite{JZSLH24}).
Results show that our \soft outperforms them, especially when identifying the further fine-tuned derivatives of the candidate model.
For example, Llama-specific~\cite{TLIMLLRGHARJGL23} \soft achieves a target response rate (TRR) of 0.98 on the further fine-tuned Llama model, while that of IF is 0.00.
Additionally, our \soft achieves an AUC of 1.00 for accurate detection of Llama, whereas TRAP achieves only 0.67.
These findings showcase the effectiveness of our method.
We further investigate the {adversarial robustness} and stealthiness of \soft through a more detailed analysis.
We also explore the influence of different token numbers and different copyright tokens, exhibiting the flexibility of our \soft.

\mypara{Contributions}
Overall, this study aims to ignite a crucial discussion on enhancing model transparency and accountability in the era of LLMs.
Our contributions are as follows:
\begin{itemize}
\item We propose \soft, a lightweight LLM identification method that proactively plugs the adversarial token embeddings into the model's token-embedding layer for further identification, without modifying transformer layers or degrading the model’s utility.
\item Further, we theoretically prove that our \soft overcomes some limitations in existing LLM identification methods.
\item Extensive experiments on three candidate models and 32 suspect models show that our \soft outperforms both state-of-the-art watermarking and fingerprinting methods.
\item Our ablation study illustrates the robustness, stealthiness, and flexibility of our method, offering opportunities for user customization.
\end{itemize}

\begin{figure*}[!t]
\centering
\includegraphics[width=1.6\columnwidth]{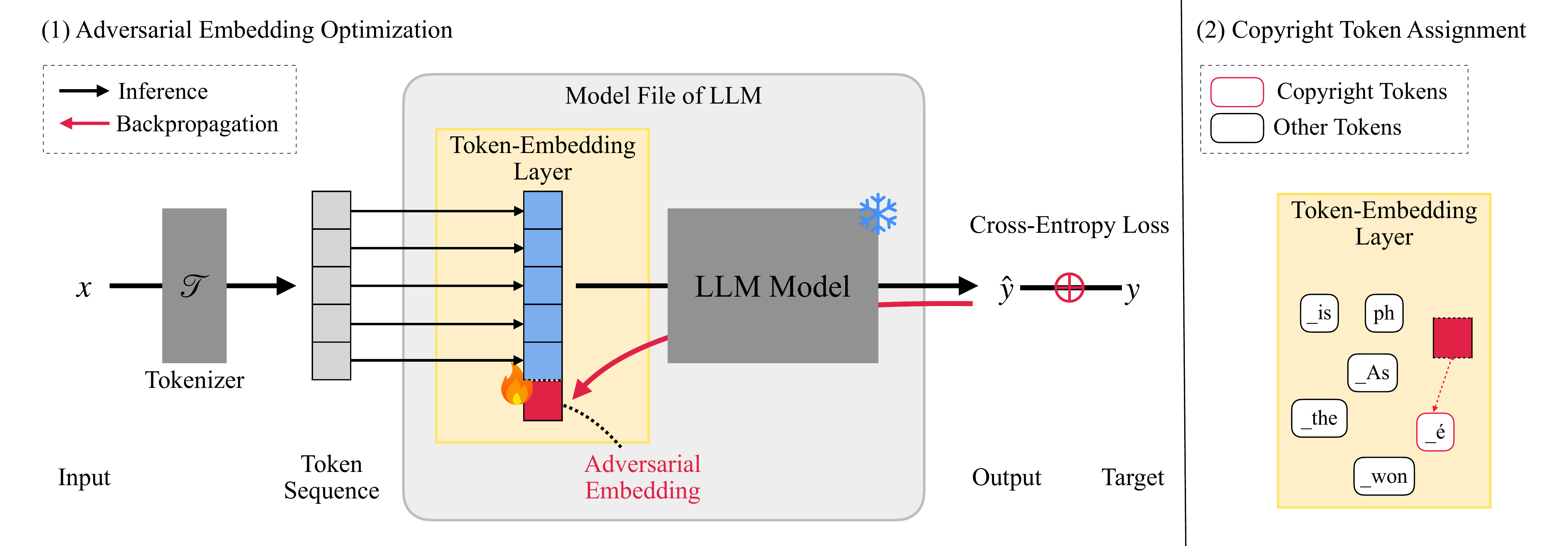}
\caption{Framework of \soft.
We first optimize adversarial embeddings on the tailored loss with the LLM frozen, then assign copyright tokens to adversarial embeddings by plugging them into the model's token-embedding layer.
}
\label{figure:framework}
\end{figure*}

\section{Related Work}

{
Based on whether the method requires embedding signals into the model proactively~\cite{CJG21, XYWLZZZXKLH25, YRLYYWL26}, we categorize existing LLM identification methods into model watermarking and fingerprinting.
We summarize the differences between the two types in \Cref{table:method_type}.
}

\mypara{Model Watermarking}
Different from text watermarking that aims to trace the LLM-generated content~\cite{KGWKMG23, YAJK23, LPHMW24, ZALW24}, model watermarking is a representative LLM identification method~\cite{GHZCH22, LCLDZL23, YBZS24, XWMKXC24, CGMR24}, most of which proactively fine-tune the LLM in a backdooring way to enable traceability by the model owner.
However, either full fine-tuning or LoRA fine-tuning~\cite{HSWALWWC22} will inevitably adjust the model's behavior, raising concerns about utility.
Moreover, when the model owner seeks to further customize the model, the watermark is prone to being eliminated after fine-tuning.

\mypara{Fingerprinting}
Fingerprinting refers to the task that aims to identify the origin of an LLM~\cite{DCS21, ZZWL23, MSSA24, GULYO24, JZSLH24} without requiring the model to be modified before release.
Existing fingerprinting methods have two categories, i.e., parameter-based and query-based.
Parameter-based fingerprinting requires white-box access to the suspect model.
For example, Zeng et al.~\cite{ZZWL23} propose a human-readable fingerprint that encodes invariant terms in the model's weight to a natural image.
Query-based fingerprinting, instead, focuses on querying the model with well-crafted inputs and identifying based on the responses.
Some works~\cite{DCS21, MSSA24} extract the characteristics and distribution from the LLM's responses to some natural text inputs.
Adversarial example (AE)-based methods, such as TRAP~\cite{GULYO24} and ProFlingo~\cite{JZSLH24}, optimize adversarial text inputs to query the LLM for target responses.

\section{Threat Model}

In the era of LLMs, \emph{adversarial third parties} may not possess the technology, computational resources, or sufficient time and funding.
Thus, they intend to deliberately deploy LLMs without proper authorization, such as deploying unauthorized releases of open-source models or leaks of proprietary black-box services~\cite{GULYO24, JZSLH24}.
The adversary can intentionally hide the copyright or trademark of the LLM and dishonestly claim that it is their own proprietary LLM, violating most LLMs' licenses and agreements, which we have summarized in \Cref{appendix:license}.
This inevitably accelerates unfair competition and blocks potential innovation.
On the contrary, the \emph{defender} can be the LLM owner or a third-party inspector ensuring the proper enforcement of intellectual property laws and regulations in the market.
They aim to identify whether the black-box suspect model is derived from the candidate model under their control.

We argue that this identification task is non-trivial in reality.
The defender often has black-box only access to the suspect model, as unauthorized services normally expose a query-only API to users.
The suspect LLM may be further customized, for example, fine-tuning on specific datasets~\cite{RGGSGTALRRKEBBCGXDCATMUSS23, TBLRRBHWFHSSRW23, IWPLPDJWSBH23, MCMCSACRJAPZRKA23} or being instructed by certain system prompts~\cite{PromptDB, introducing_gptstore}.
Such customization, while not intended, could help the adversary avoid identification by the defender.
To further circumvent identification as the candidate model, the adversary may use specific prompts to respond to any queries about their true origin.
They may also rebrand the LLM and rename the parameters.
Additionally, they could rearrange the model's weights, significantly changing the direction of parameter vectors without affecting the model's performance~\cite{ZZWL23}.

\section{\soft}
\label{section:soft}

\subsection{Preliminary}
\label{section:preliminary}

Large language models (LLMs) typically consist of a tokenizer and a model.
The tokenizer maps between text $x$ and tokens $\{t_1, \ldots, t_n\}$ (denoted as $t_{1:n}$ for simplicity).
With the input tokens $t_{1:n}$, the token-embedding layer of the model converts them into \emph{token embeddings}: $g(t_{1:n}) = \mathbf{e}_{1:n}, \mathbf{e}_i \in \mathbb{R}^{d}$.
The remained transformer layers, parameterized by $\Phi$, then calculates the probability of producing the next token $p_{\Phi}(t_{n+1} | \textbf{e}_{1:n})$.

Overall, the LLM determines the joint probability of generating a series of tokens $t_{n+1:n+m}$ autoregressively: $p_{\Phi}(t_{n+1:n+m}|\mathbf{e}_{1:n})=\prod^{m}_{i=1}p_{\Phi}(t_{n+i}|\mathbf{e}_{1:n+i-1})$.
The typical training process of LLM leverages the negative log-likelihood (NLL) loss:
\begin{equation}
\label{equation:nll_loss}
\begin{aligned}
\mathcal{L}(y, x) &=\mathcal{L}^{\text{e}}(y, \textbf{e}_{1:n})= - \sum_{i=1}^{m}\log p(t_{n+i}| \textbf{e}_{1:n+i-1}),
\end{aligned}
\end{equation}
where $\mathcal{L}^{\text{e}}(\cdot, \cdot)$ denotes the loss function that takes $y$ and the token embeddings of $x$ as input.

\subsection{\soft Approach}
\label{section:method}

\mypara{Method Overview}
In this work, we propose \soft, an efficient and effective LLM identification method that proactively \underline{\texttt{Plug}}s the \underline{\texttt{A}}dversarial token \underline{\texttt{E}}mbeddings into the LLM's token-embedding layer for identification, without modifying its main transformer layers $\Phi$.
Specifically, our method has two stages, i.e., \emph{adversarial embedding optimization} and \emph{copyright token assignment}, as shown in \Cref{figure:framework}.
We detail the two stages below and address potential concerns.

\mypara{Stage I. Adversarial Embedding Optimization}
Instead of optimizing the discrete adversarial tokens, we optimize a sequence of $k$ adversarial token embeddings $\mathbf{e}^a$.
Such adaptation makes the optimization process differentiable, largely increasing the efficiency.
The loss function becomes:
\begin{equation}
\label{equation:loss_method}
\begin{aligned}
\mathcal{L}_{\soft}(y, x, \mathbf{e}^a) = \sum_{h\in H}\mathcal{L}^{\text{e}}(y, \mathbf{insert}_{z_h}(\mathbf{e}^{{h(x)}}, \mathbf{e}^a)),
\end{aligned}
\end{equation}
where $h(x)$ represents encoding user input $x$ with the chat template $h \in H$, a set containing different chat templates with various system prompts.
$\mathbf{insert}_{z_h}(s_1, s_2)$ denotes inserting sequence $s_2$ into $s_1$ at the position $z_h$.
The position $z_h$ is adaptable for different chat templates.
This makes \soft more generalizable to complex settings, such as system prompts or instructional templates.
Based on the frozen LLM, we optimize the adversarial embeddings $\mathbf{e}^a$ to minimize the loss $\mathcal{L}_{\soft}$:
\begin{equation}
\begin{aligned}
\arg\min_{\mathbf{e}^a}\mathcal{L}_{\soft}(y, x, \mathbf{e}^a), (x,y) \in Q,
\end{aligned}
\end{equation}
where we optimize a \emph{universal} sequence of adversarial embeddings $\mathbf{e}^a$ for all queries in set $Q$.

\mypara{Stage II. Copyright Token Assignment}
\soft then assigns the obtained $k$ adversarial embeddings $\mathbf{e}^a$ to customized \tokens in the LLM's token-embedding layer $g$.
The \tokens refer to the tokens being used to distinguish the LLM's identity afterwards.
Specifically, \soft allows users to design their own \tokens, which can be updated and modified based on users' needs.
Generally, the tokens should be rarely used or under-trained~\cite{SolidGoldMagikarp, LB24}, such as ``SolidGoldMagikarp'' and ``cf,'' so that the model's utility is preserved for most cases.
We also show that the user can design more unique tokens, such as ``mkahg,'' or more complex ``h?!4t2rvAEVhC>K,'' that are not contained in the vocabulary.
In such circumstances, we need to add these newly created tokens into the tokenizer, which leads to a trade-off between the rareness of tokens and the stealthiness of \soft.
With the selected \tokens, \soft sets their mapped token embeddings in the token-embedding layer $g$ as the optimized adversarial embeddings, accordingly.

\subsection{\soft v.s. Model Watermarking}

{Compared with most model watermarking methods that require fine-tuning the LLM (including both token embeddings in $g$ and the transformer layers ${\Phi}$) via backdooring~\cite{XWMKXC24}, \soft only modifies the token-embedding layer $g$.
Thus, the model's utility is well-preserved: in practice, only 1-3 copyright tokens have specific behaviors, leaving the remaining 32K-216K tokens to act as originals.
Additionally, this plug-in process is more adaptable and convenient.}
Imagine an LLM already has several fine-tuned derivatives before release.
The model owner wants to watermark this LLM for copyright protection via backdooring.
Then, every existing derivative of this LLM is required to be watermarked as well, which consumes both time and computational resources.
Instead, our \soft only requires plugging the optimized adversarial embeddings into the existing derivatives.
That is, they do not need to optimize again.
Experiments in \Cref{section:experiment_compare_fingerprint} show that our method remains effective under such scenarios.
We further discuss remaining challenges and our vision in \Cref{appendix:challenges}.

\subsection{\soft v.s. Fingerprinting}

Query-based fingerprinting, such as TRAP~\cite{GULYO24} and ProFlingo~\cite{JZSLH24}, is the mainstream fingerprinting approach to identifying the origin of a black-box LLM.
They share a similar goal with our \soft, i.e., leveraging adversarial prompts to guide the candidate LLM to generate target outputs.
As detailed in \Cref{appendix:theory_fingerprint}, we summarize these query-based methods with a unified loss:
\begin{equation}
\label{equation:unified_loss}
\begin{aligned}
\mathcal{L}_{\text{AE}}(y, x, a) = \sum_{h\in H}\mathcal{L}(y, \mathbf{insert}_{z_h}(h(x), a)).
\end{aligned}
\end{equation}
Specifically, the unified loss $\mathcal{L}_{\text{AE}}$ can be adapted to ProFlingo~\cite{JZSLH24} when we insert the adversarial example $a$ as input $x$'s prefix, i.e., let $\mathbf{insert}_{z_h}(h(x), a) := h(a\| x)$, where $\|$ denotes sequence concatenation.
Likewise, this unified loss can be translated to TRAP if we limit to \emph{one} default chat template $H = \{h\}$ and append $a$ after $x$, i.e., $\mathbf{insert}_{z_h}(h(x), a):= {h(x\| a)}$.

However, they can only optimize on discrete tokens within the vocabulary, the size of which is much smaller than the embedding space in $\mathbb{R}^d$ where $d$ denotes the token embedding dimensions.
This inevitably limits the search space of query-based fingerprinting and bounds their capability.
By comparing our $\mathcal{L}_{\soft}$ in \Cref{equation:loss_method} and the unified loss $\mathcal{L}_{\text{AE}}$ in \Cref{equation:unified_loss}, we hereby theoretically prove that \soft is more likely to reach the optimal compared to query-based fingerprinting.

\begin{hypothesis}
\label{hypothesis:injective}
For most LLMs, we assume their embedding layer $g: V^n \rightarrow \mathbb{R}^{d\times n}$ is injective, where $d$ is the dimension of the token embedding and $n$ is the length of the token sequence.
\end{hypothesis}
We believe this hypothesis holds as the size of the vocabulary $V$ is much smaller than the embedding space $\mathbb{R}^{d}$, where $d$ denotes the token embedding dimensions.
{To empirically verify this hypothesis, we checked all 30 LLMs used in our experiments (see \Cref{section:settings}).
Specifically, for each LLM, we calculate the cosine similarity between every pair of token embeddings.
We considered two embeddings to be similar if their cosine similarity exceeded 0.99.
The results of our experiments demonstrate that, across all the LLMs tested, the token embeddings were pairwise distinct.
The empirical evidence supports that our hypothesis holds for most LLMs in practice.}

Based on the hypothesis, we further prove that the inverse function $q$ of the embedding layer $g$ is non-differentiable, so that:
\begin{lemma}
\label{lemma:loss_closer_to_optimal}
Given $q$ is non-differentiable, $\mathcal{L}_{\soft}$ is prone to reach closer or equal to the optimal solution compared to the unified loss $\mathcal{L}_{\text{AE}}$.
\end{lemma}

Based on \Cref{lemma:loss_closer_to_optimal}, we claim that the optimization of our \soft is more likely to reach the optimal compared with current query-based methods that adopt the unified loss $\mathcal{L}_{\text{AE}}$.
This indicates that the adversarial examples optimized by our \soft can extract more information from the LLM, making it more reliable for identification, which is proven experimentally in \Cref{section:experiment_compare_fingerprint}.
Due to the page limit, the full proof can be found in \Cref{appendix:optimal_proof}.

\section{Experiments}

\begin{figure*}[!t]
\centering
\includegraphics[width=2\columnwidth]{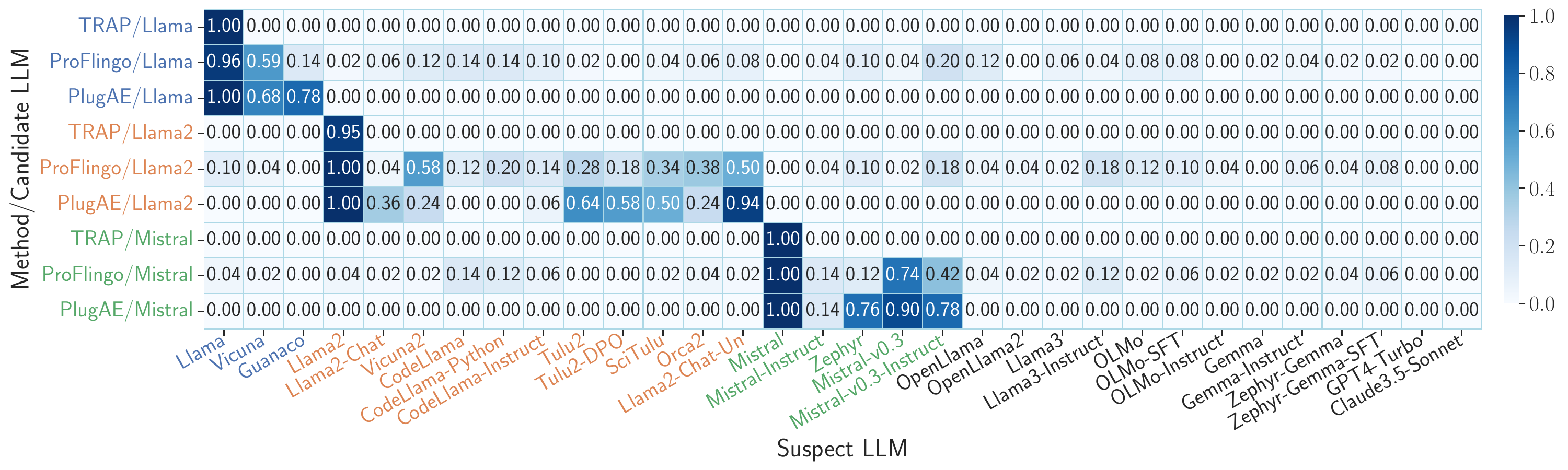}
\caption{Performance (TRR) of different methods for three candidate LLMs on 32 suspect LLMs.
A higher TRR indicates that the suspect LLM is more probable to be derived from the candidate LLM.
Note that the same color indicates that the suspect LLM is derived from the candidate LLM.
}
\label{figure:fingerprint_results}
\end{figure*}

\subsection{Experimental Settings}
\label{section:settings}

\mypara{Models}
We utilize three LLMs as \emph{candidate models}, including Llama~\cite{TLIMLLRGHARJGL23}, Llama2~\cite{TMSAABBBBBBBCCCEFFFFGGGHHHIKKKKKKLLLLLMMMMMNPRRSSSSSTTTWKXYZZFKNRSES23}, and Mistral~\cite{JSMBCCBLLSLLSSLWLS23}.
We have white-box access to them and aim to identify whether they are the origin of the suspect models.
Then, we collect 30 open-source LLMs and two real-world black-box APIs as \emph{suspect models}.
Specifically, models derived from the candidate models are released both officially (e.g., Llama2-Chat~\cite{TMSAABBBBBBBCCCEFFFFGGGHHHIKKKKKKLLLLLMMMMMNPRRSSSSSTTTWKXYZZFKNRSES23} developed by Meta) and by other institutes (e.g., Vicuna~\cite{Vicuna} (fine-tuned on Llama) from LMSYS Org).
We also collect other model families, such as OLMo~\cite{GBWBKTJIMWAAACCDEGHKMMMNNPPRSSSSSWDLRZDLSSH24}, and two black-box APIs, i.e., GPT4-Turbo~\cite{O23} and Claude3.5-Sonnet~\cite{claude3.5}.
All LLMs used are publicly available.
We do not use any unauthorized LLMs in our study, thereby adhering to ethical standards and legal compliance.
Model details can be found in \Cref{appendix:settings}.

\mypara{Baselines}
We compare with both state-of-the-art model watermarking, Instructional Fingerprint (IF)~\cite{XWMKXC24}, and fingerprinting methods, TRAP~\cite{GULYO24} and ProFlingo~\cite{JZSLH24}.
We follow their default settings.
To compare with fingerprinting, we conduct candidate model-specific identifications on 32 suspect models.
For model watermarking, we mock the customization process on the watermarked LLMs to evaluate the methods' robustness.
Specifically, we fine-tune the watermarked/plugged model on CodeSearchNet (CSN)~\cite{HWGAB19}, a widely used code dataset that contains about six million functions spanning six programming languages.

\mypara{Configurations}
In experiments, we place the adversarial embeddings as the prefix of the input query and use the query set of ProFlingo.
We set the number of adversarial token embeddings $k$ to 1 and set the \token to a randomly generated 5-digit string ``mkahg.''
The learning rate is set to 0.1 with the Adam optimizer for running 30 epochs.
We use the top-p of 1.0 and temperature of 1.0 as default settings when generating responses.
In evaluation, we use the prompt template suggested in its repository for each model if available.
Otherwise, we default to the zero-shot prompt template from FastChat~\cite{ZCSZWZLLLXZGS23}.
For fine-tuning the watermarked LLMs, we adhere to the training parameters of Alpaca~\cite{stanford_alpaca} and run for 15,000 steps on four NVIDIA A100 GPUs with 40 GB of memory.
Then, we evaluate the utility of the fine-tuned LLMs on HellaSwag~\cite{ZHBFC19} and MMLU~\cite{HBBZMSS21} under the zero-shot setting.

\mypara{Evaluation Metric}
We use the target response rate (TRR) for evaluation, defined as the ratio of the target outputs among all responses, i.e., $\text{TRR}=\frac{\# \text{Correct target outputs}}{\#\text{Total outputs}}$.
A higher TRR suggests stronger evidence that the suspect model is derived from the candidate model.
Considering this as a binary classification task, we also report the classification accuracy and calculate the ROC/AUC curve when evaluating on 32 suspect models.

\subsection{Experiment: \soft v.s. Model Watermarking}
\label{section:experiment_compare_watermark}

\begin{table}[!t]
\centering
\setlength{\tabcolsep}{1.1mm}
\scalebox{0.80}{
\begin{tabular}{lccc|ccc|ccc}
\toprule
\multirow{2}{*}{Method} &\multicolumn{3}{c}{Llama}  &\multicolumn{3}{c}{Llama2}&\multicolumn{3}{c}{Mistral}\\
& \Circle & \CIRCLE & \LEFTcircle & \Circle & \CIRCLE & \LEFTcircle & \Circle & \CIRCLE & \LEFTcircle \\
\midrule
IF&0.00&1.00&0.00&0.00&1.00&0.88&0.00&1.00&0.25\\
\soft &0.00&1.00&0.98&0.00&1.00&0.98&0.00&1.00&0.78\\
\bottomrule
\end{tabular}
}
\caption{
TRR of model watermarking and our \soft on three different candidate LLMs.
\Circle\xspace denotes the original LLMs without being modified; \CIRCLE\xspace represents the modified LLMs either being watermarked or modified by \soft;  \LEFTcircle\xspace depicts modified LLMs that are further fine-tuned on the CSN dataset.
}
\label{table:performance_soft_csn}
\end{table}

Experimental results of our \soft depicted in \Cref{table:performance_soft_csn} show its effectiveness in identifying and robustness against fine-tuning.
We observe that both \soft and IF achieve a TRR of 1.00 when we identify the suspect model without further fine-tuning.
Then we mock the customization process and fine-tune the modified LLM on CSN.
We observe that \soft maintains a high TRR of 0.98 on both fine-tuned Llama and Llama2 models, which outperforms IF by 0.98 and 0.10, respectively.
This indicates that the performance of backdoor-based model watermarking relies on the customization dataset and shows that our method is more robust against fine-tuning compared with model watermarking.

The utility of LLMs is examined in \Cref{appendix:utility}.
The results show that fine-tuning and model watermarking do not affect the model's utility.
We also find that our \soft achieves exactly the same performance as the original model under the same random seed.
This makes sense as we only insert several adversarial token embeddings and do not need to modify the model $\Phi$.

Besides, we investigate the influence of fine-tuning on the plugged adversarial embeddings, i.e., the token embedding of the \token.
Specifically, we calculate the cosine similarity between the original adversarial embeddings and them in the fine-tuned LLM.
Results show that the adversarial embeddings remain the same after fine-tuning in all three candidate LLMs, i.e., all cosine similarities are 1.00.
This is intuitive as we assume the \token is defined by the model owner, which should be rare and hard to guess.
This finding provides additional support for our earlier observation, i.e., \soft is more robust against fine-tuning than backdoor-based model watermarking.

\subsection{Experiment: \soft v.s. Fingerprinting}
\label{section:experiment_compare_fingerprint}

We simulate the fine-tuning scenario in the real world by plugging the \token of the candidate model into their existing derivatives.
For example, we add the \token of Llama2 into Llama2-Chat and the other nine derivatives.
This is practical as the token embedding of \token is robust against fine-tuning, as shown in \Cref{section:experiment_compare_watermark}.
Based on this, we compare our \soft with two state-of-the-art fingerprinting methods, TRAP~\cite{GULYO24} and ProFlingo~\cite{JZSLH24} on 32 suspect models.

\Cref{figure:fingerprint_results} shows the TRRs of each method for identifying three candidate LLMs on 32 suspect models.
We observe that our \soft outperforms fingerprinting in detecting both non-derivatives and derivatives.
For example, given a candidate model with \soft, the TRRs of suspect models that are non-derivatives are 0.00.
This is because those models do not have the corresponding \token or the adversarial embeddings.
In this sense, when we query the model with \tokens, non-derivative models would not respond with the target answer.
Instead, ProFlingo tends to overestimate some non-derivatives as the candidate LLM's derivative.
For example, Llama-based ProFlingo achieves a TRR of 0.20 on Mistral-v0.3-Instruct, while it only achieves 0.14 on Guanaco, the fine-tuned version of Llama.
For derivatives, our \soft achieves generally higher TRRs than ProFlingo, while TRAP fails to distinguish fine-tuned versions.
For example, based on Llama, \soft achieves a TRR of 0.78 on the derivative Guanaco, while TRAP and ProFlingo achieve 0.00 and 0.14, respectively.
Similarly, Mistral-based \soft achieves a TRR of 0.76 on Zephyr, while TRAP and ProFlingo only have a TRR of 0.00 and 0.12.
This indicates that our \soft illustrates more confidence in distinguishing between the derivative and non-derivative LLMs.

Considering this as a binary classification task, we calculate the ROC/AUC curve for each method on three LLMs.
As shown in \Cref{figure:roc_llama_7b_alone}, our \soft achieves an AUC of 1.00 on Llama, outperforming the other two methods.
This further shows the effectiveness of our method.
The results of the other two models can be found in \Cref{appendix:roc}

\subsection{Stealthiness}
\label{section:stealthiness}

\begin{figure}[!t]
\centering
\includegraphics[width=0.75\columnwidth]{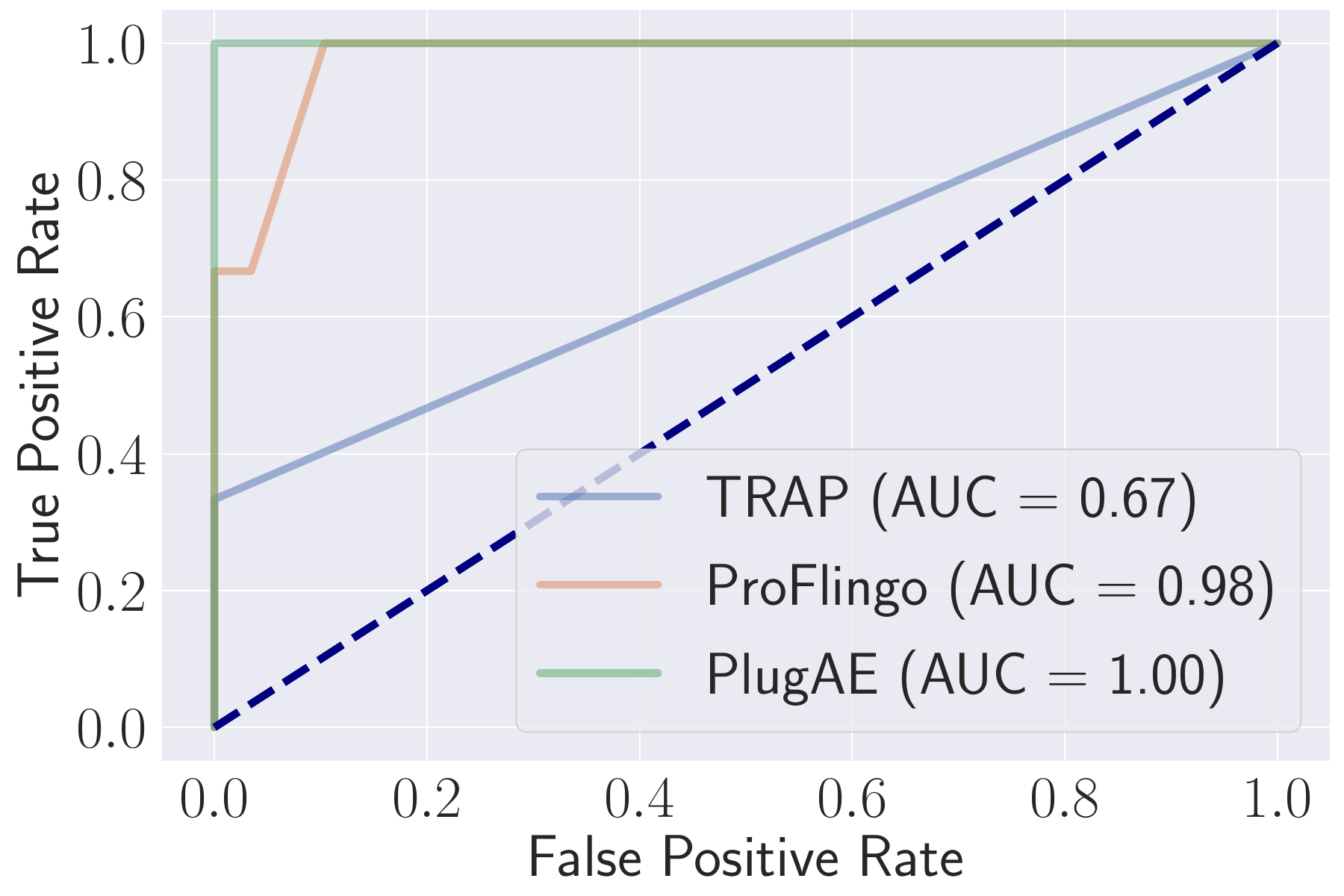}
\caption{ROC curve of different methods on Llama.
}
\label{figure:roc_llama_7b_alone}
\end{figure}

\begin{figure}[!t]
\centering
\includegraphics[width=0.75\columnwidth]{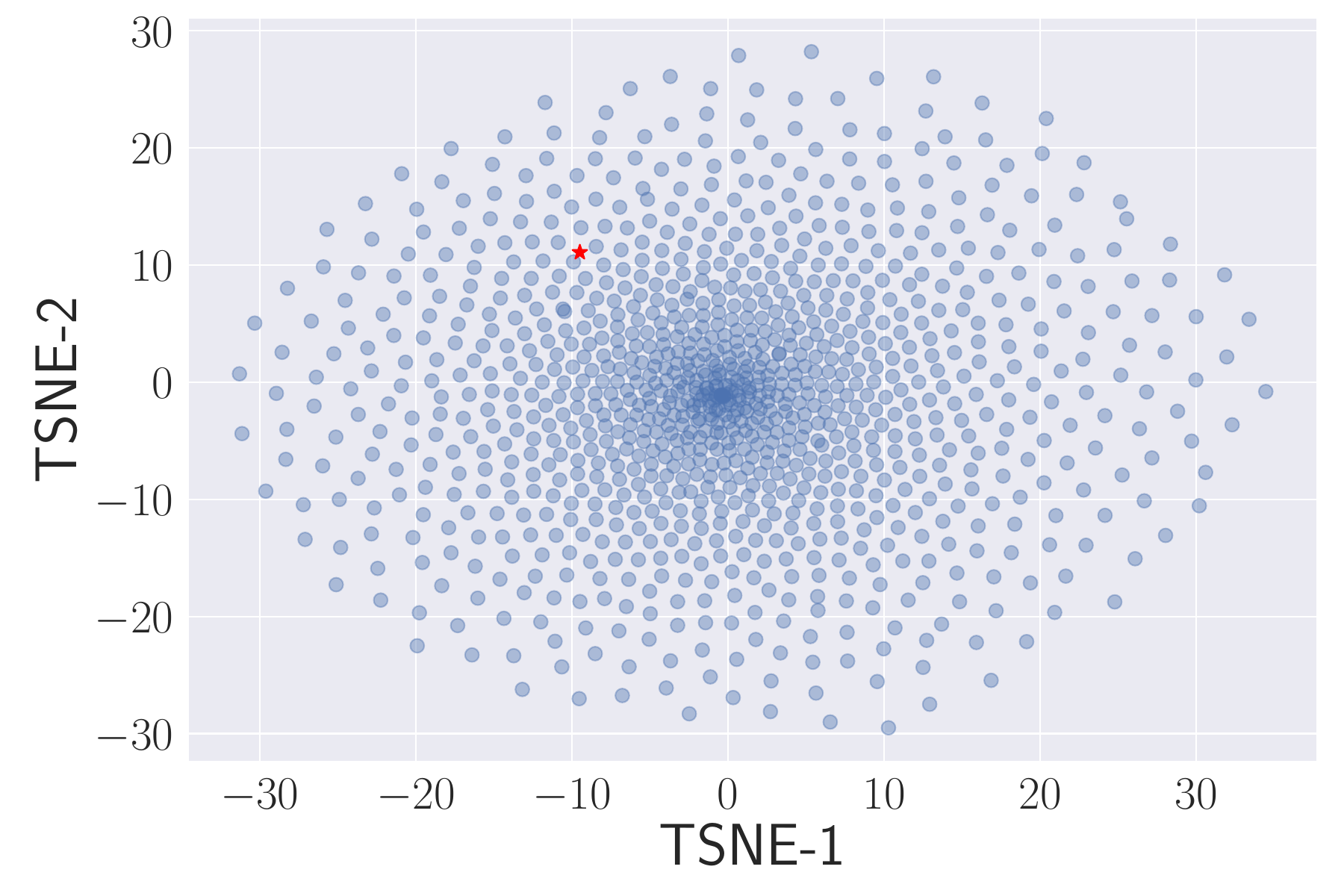}
\caption{T-SNE visualization of adversarial embedding (red) and 1,000 random token embeddings in Llama.
}
\label{figure:tsne_llama_7b_alone}
\end{figure}

We investigate the stealthiness of our \soft considering three aspects, i.e., \tokens, the adversarial token embeddings, and the query set $Q$.

\mypara{Copyright Tokens}
We give users a high degree of flexibility in customizing their own \tokens and provide some suggestions.
Generally, the \tokens should be rarely used so that they will not be triggered for normal inputs.
We also mentioned that the users can add new tokens to the tokenizer, as it is natural for LLMs to have their own tokenizer~\cite{TMSAABBBBBBBCCCEFFFFGGGHHHIKKKKKKLLLLLMMMMMNPRRSSSSSTTTWKXYZZFKNRSES23, JSMBCCBLLSLLSSLWLS23, T24}, i.e., there is no unified tokenizer that works for all LLMs.
The vocabulary size also makes it hard to check all tokens, e.g., there are 32,000 tokens in Llama2 and 256,000 tokens in Gemma~\cite{T24}.
We can even remove the same number of rare or under-trained tokens to keep the vocabulary size.
Further, when $k > 1$, we can use a cryptographically secure pseudorandom number generator (CSPRNG)~\cite{SGA11} to select the \token sequence from a set of rare tokens.
All these strategies make it difficult to perceive \tokens.

\mypara{Adversarial Token Embeddings}
We further explore whether we can detect \tokens by analyzing token embeddings.
For each candidate LLM, we randomly sample 1,000 token embeddings and use t-SNE to visualize them and the corresponding adversarial embedding in this model.
Due to the page limit, we exhibit the results of Llama (see \Cref{figure:tsne_llama_7b_alone}).
The visualization results on all three LLMs can be found in \Cref{figure:tsne_embeddings} in \Cref{appendix:tsne}.
We observe that the adversarial embedding for Llama is within the circle of its token embeddings, i.e., it is not an outlier.
In this sense, it is challenging to detect the adversarial embedding using outlier detection methods, exhibiting the stealthiness of our \soft.

\mypara{Query Set}
We allow users to customize their query set $Q$, which further complicates the detection of \soft.
Even if the attacker knows the \tokens, it is challenging for them to deduce our exact approach through reverse engineering.
We further show that our \soft does not rely on any specific query set.
Specifically, we use ChatGPT~\cite{chatgpt} to generate a similar query set with 50 questions.
Then we conduct \soft on Llama2 and evaluate on different LLMs.
The results shown in \Cref{table:query_set} demonstrate the robustness of our method using different query sets.

\begin{table}[!t]
\centering
\setlength{\tabcolsep}{1.1mm}
\scalebox{0.80}{
\begin{tabular}{lcc}
\toprule
Suspect Model & \soft & \soft (New Query Set)\\
\midrule
Llama2 & 1.00 & 1.00 \\
Llama & 0.00 & 0.02 \\
Mistral & 0.00 & 0.02 \\
OpenLlama & 0.00 & 0.02 \\
OLMo & 0.00 & 0.02 \\
\bottomrule
\end{tabular}
}
\caption{
TRR of Llama2-specific \soft with a new query set on different suspect LLMs.
}
\label{table:query_set}
\end{table}

\subsection{Ablation Study}
\label{section:ablation_study}

\begin{table}[!t]
\centering
\setlength{\tabcolsep}{1.1mm}
\scalebox{0.80}{
\begin{tabular}{lccccc}
\toprule
Suspect Model & Empty & Basic & Helpful & Alpaca & ChatGPT \\
\midrule
Llama2 & 0.88 & 1.00 & 0.98 & 1.00 & 1.00 \\
Llama & 0.00 & 0.00 & 0.00 & 0.00 & 0.00 \\
Mistral & 0.02 & 0.00 & 0.04 & 0.00 & 0.00 \\
OpenLlama & 0.02 & 0.00 & 0.00 & 0.00 & 0.00 \\
OLMo & 0.00 & 0.00 & 0.00 & 0.02 & 0.00 \\
\bottomrule
\end{tabular}
}
\caption{
TRR of Llama2-specific \soft with various system prompts on different suspect LLMs.
}
\label{table:influence_system_prompt}
\end{table}

\mypara{Copyright Tokens}
In the main experiments, \soft takes a random 5-digit string as the \token.
To investigate the influence of different copyright token choices, we experiment with the other two different \tokens.
The first is that we randomly generate a 20-digit-long string as the \token, i.e., ``c*zNFj.7\}4E5\textasciitilde hHAq9Nh.''
The other is ``cf,'' which exists in Mistral-7b's default vocabulary but not in Llama or Llama2's vocabulary.
In the experiments, we got the same results using different \tokens.
This is intuitive, as different \tokens correspond to the same token embedding.

\mypara{Chat Templates}
{We follow the official chat template for different LLMs, so that the derivatives use different chat templates and system prompts from the original LLM.
For example, the chat template of CodeLlama is different from that of the base Llama2 model.
Results shown in \Cref{figure:fingerprint_results} demonstrate the robustness of our \soft regarding different prompt wrappers.}

\mypara{System Prompts}
{We also evaluate \soft on LLMs with various system prompts.
We use five different system prompts: empty, basic, helpful, Alpaca, and ChatGPT style, as detailed in \Cref{table:system_prompts} in the appendix.
We apply these system prompts to the suspect LLMs and evaluate the Llama2-specific PlugAE on these suspect LLMs.
Results are shown in \Cref{table:influence_system_prompt}.
We observe that although there are some variants with different system prompts, we can still identify the derivatives and the non-derivatives of Llama2 model.
This indicates that our \soft does not rely on fragile system prompts.}

\mypara{Numbers of Tokens}
To investigate the influence of the number of \tokens, we further experiment with the number of tokens of 2, 3, and 5 on the Llama2 model.
For \token selection, we randomly generate a 5-digit string for each token embedding.
Results in \Cref{figure:influence_token_number} show that a larger token number could achieve a better TRR on the derivatives.
For instance, the TRR on SciTulu increases from 0.50 to 0.94 when the token number increases to 5.
However, when the token number is too large, i.e., 5, the TRRs on the derivatives would even drop.
For example, the TRR on Orca2 drops to 0.52 when the token number is 5, while it is 0.72 when the token number is 3.
Note that the TRRs for non-derivatives remain 0.

\begin{figure}[!t]
\centering
\includegraphics[width=0.75\columnwidth]{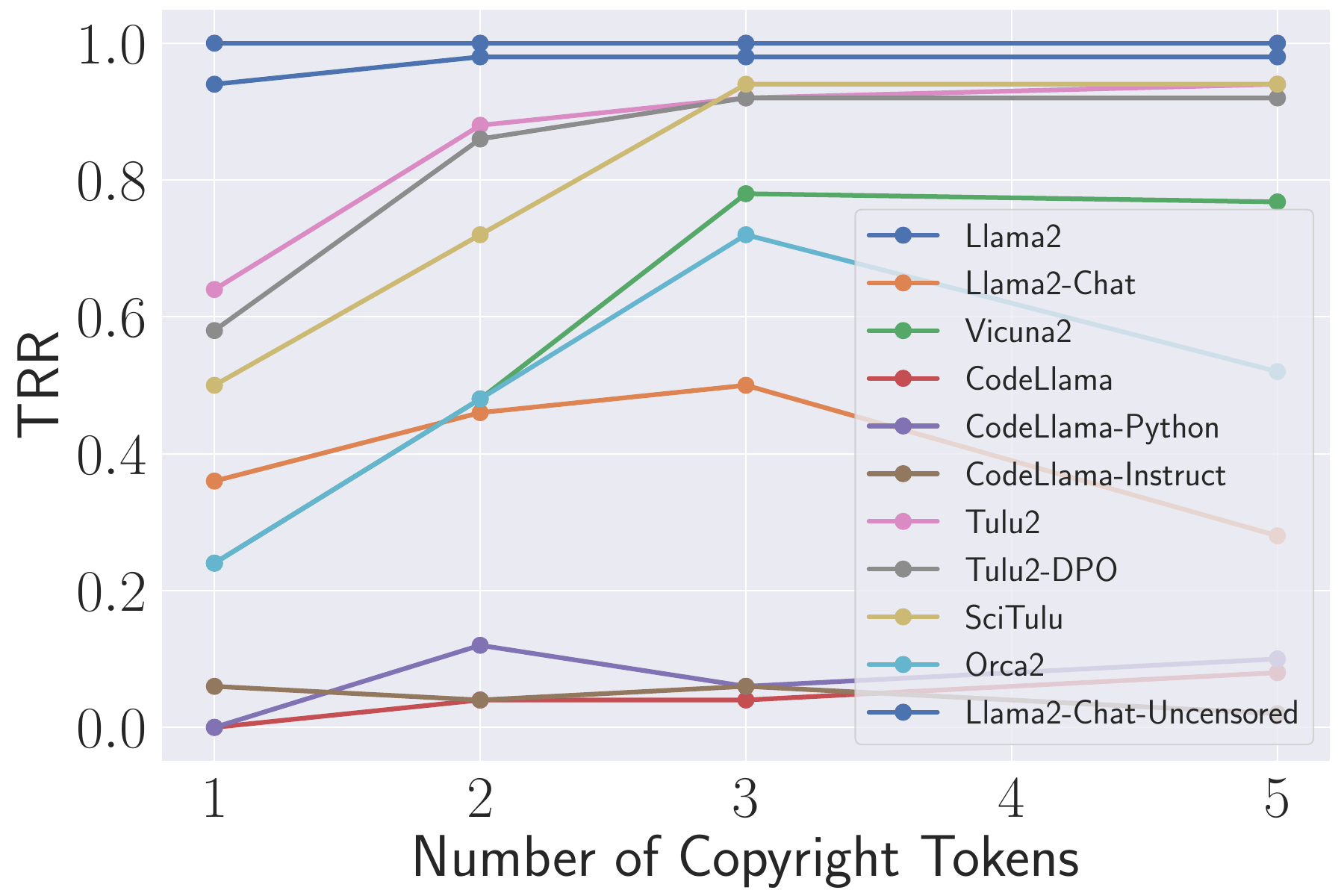}
\caption{Influence of the number of tokens on Llama2.
TRRs for non-derivatives are consistently 0.0.
}
\label{figure:influence_token_number}
\end{figure}

\begin{table}[!t]
\centering
\setlength{\tabcolsep}{1.1mm}
\scalebox{0.80}{
\begin{tabular}{lcc}
\toprule
Suspect LLM & Llama3.1 & Mistral-24B \\
\midrule
Llama3.1 & 0.94 & 0.00 \\
Llama3.1-Chat & 0.82 & 0.00 \\
Mistral-24B & 0.00 & 0.92 \\
Mistral-24B-Chat & 0.00 & 0.80 \\
Mistral-24B-Chatml & 0.00 & 0.86 \\
GPT5 & 0.00 & 0.00 \\
Gemini2.5 & 0.00 & 0.00 \\
Claude4.5 & 0.00 & 0.00 \\
\bottomrule
\end{tabular}
}
\caption{
TRR of Llama3.1- and Mistral-24B-specific \soft on eight suspect models.
}
\label{table:other_model_performance}
\end{table}

\mypara{More Advanced Models}
To demonstrate the generalizability of \soft, we extended experiments to larger and more frontier LLMs.
We use Llama3.1~\cite{llama3_1_blog} (meta-llama/Llama-3.1-8B, released in 2024) and Mistral-24B~\cite{mistral_small_3} (mistralai/Mistral-Small-24B-Instruct-2501, released in 2025) as candidate models.
And we take Llama3.1, Llama3.1-Chat (meta-llama/Llama-3.1-8B-Instruct), Mistral-24B, Mistral-24B-Chat (mistralai/Mistral-Small-3.1-24B-Instruct-2503), and Mistral-24B-Chatml (IntervitensInc/Mistral-Small-24B-Instruct-2501-chatml), and three frontier APIs as suspect models.
The APIs include GPT5~\cite{gpt_5} (gpt-5-mini-2025-08-07), Gemini2.5~\cite{gemini_checkpoint} (gemini-2.5-flash), and Claude4.5~\cite{claude4} (claude-sonnet-4-5-20250929).
As shown in \Cref{table:other_model_performance}, \soft consistently achieves high TRR across derivatives, while unrelated models and APIs yield 0 TRR.
This confirms strong detectability of our \soft.

\section{Conclusion}
\label{section:conclusion}

In this paper, we address the challenge of identifying the origin of black-box LLMs, motivated by the increasing risks of copyright infringement and unfair competition.
To this end, we introduce \soft, a simple yet effective identification framework that integrates adversarial token embeddings into the token-embedding layer of LLMs without modifying the model's {main transformer layers}.
Our theoretical analysis reveals the inner limitations of both fingerprinting and model watermarking methods, and positions \soft as a robust alternative.
Extensive experiments on 30 LLMs and two real-world LLM API services empirically validate the effectiveness of \soft, consistently outperforming state-of-the-art LLM identification techniques.
We further exploit the stealthiness of our \soft from three distinct perspectives.
Finally, the ablation study under various configurations confirms the robustness and flexibility of our method.

\section*{Limitations}

While our \soft demonstrates strong effectiveness in identifying the origin of LLMs, it is a proactive method and still requires modifying the corresponding LLM before release.
This introduces the risk of the modification being detected or removed.
We leave this as a future direction.
Further, due to the project timeline, we opted for some older LLM versions rather than the most current one, while they are compatible with the LLMs used in the baselines.
{To keep updated, we have experimented with more advanced and larger models and validated the effectiveness of our method in \Cref{section:ablation_study}.
}

\section*{Ethical Consideration}

By raising awareness of the risks of the unauthorized use of LLMs, our work aims to contribute to developing a more robust and secure open-source environment and calls for a more reliable way to identify LLMs.
With theoretical analysis, we explore the inner limitations of current query-based fingerprinting methods and provide an alternative way to overcome such limitations.
We believe our work can motivate more advanced identification methods in this field.
All LLMs are specified in \Cref{appendix:settings}, and the datasets used in this research are in English.
They are either publicly available or generated by LLMs, ensuring no inclusion of personally identifiable information and eliminating user de-anonymization risks.
We do not use any unauthorized LLMs in our study, thereby adhering to ethical standards and legal compliance.

\begin{small}
\bibliographystyle{plain}
\bibliography{normal_generated_py3}
\end{small}

\cleardoublepage
\appendix
\crefalias{section}{appendix}
\crefalias{subsection}{appendix}

In the appendix, we provide details and results omitted in the main text.
\begin{itemize}
    \item \Cref{appendix:license}: LLM licenses and agreements of six mainstream LLM organizations on the market, including OpenAI, Anthropic, Meta, Ai2, Google, and Mistral AI.
    \item \Cref{appendix:theory}:
    Theoretical analysis and proof details on fingerprinting and model watermakring.
    \item \Cref{appendix:challenges}: Based on our theoretical and empirical analysis, we highlight challenges and opportunities in the task of LLM identification and copyright protection.
    \item \Cref{appendix:settings}: Details of the models and datasets we use in the experiments.
    \item \Cref{appendix:experiments}: Additional experiments, including the utility evaluation, AUC/ROC and TSNE results on three LLMs, and counterexample analysis.
\end{itemize}

\section{LLM Licenses and Agreements}
\label{appendix:license}

\begin{table*}[!t]
\centering
\setlength{\tabcolsep}{1.1mm}
\scalebox{0.80}{
\begin{threeparttable}
\begin{tabular}{lcccccc}
\toprule
\multirow{2}{*}{Usage Scenario} & \multicolumn{6}{c}{Organization} \\
& OpenAI & Anthropic & Google & Mistral AI & Meta & Ai2 \\
\midrule
Open-Source Model\tnote{1} & \xmark & \xmark & \cmark & \cmark & \cmark & \cmark \\
Query-Only API Service & \cmark & \cmark & \cmark & \cmark & \xmark & \xmark \\
Fine-Tuning API Service & \cmark & \cmark & \cmark & \cmark & \xmark & \xmark \\
\midrule
Commercial Use & \xmark & \xmark & \na & \xmark & \cmark & \xmark \\
Commercial Use (Special Permission) & \cmark & \na & \cmark & \cmark & \cmark & \cmark \\
Rebranding & \xmark & \xmark & \xmark & \xmark & \xmark & \xmark\tnote{2} \\
\bottomrule
\end{tabular}
\begin{tablenotes}
    \item[1] The open-source models include licenses not yet approved by the Open Source Initiative (\url{https://opensource.org/license}).
    \item[2] Note that Ai2's OLMo series is under Apache 2.0 license~\cite{apache_license_2_0} and allows rebranding, while there are other LLMs, such as tulu-2 series, are under the AI2 ImpACT License for Low Risk Artifacts (\url{https://huggingface.co/allenai/tulu-2-7b/blob/main/LICENSE.md}), which forbids rebranding.
\end{tablenotes}
\end{threeparttable}
}
\caption{Usage scenarios by each organization's model/API licenses.
``\na'' denotes \emph{not applicable}.
}
\label{table:licenses}
\end{table*}

We summarize the types of services provided by the most mainstream LLMs on the market, along with their usage restrictions in~\Cref{table:licenses}.
Among the six organizations we collected, OpenAI and Anthropic focus on providing API services, Meta and Ai2 release open-source models, while Google and Mistral AI have both API services and open-source models.
We further analyze the usage restrictions in their licenses.

Our analysis reveals that many impose constraints on commercial use.
Specifically, OpenAI requires businesses and developers that could benefit from their services to agree to further business terms~\cite{openai_license2024, openai_business_license2025}.
Anthropic stipulates that their services cannot be used for commercial or business purposes~\cite{anthropic_license}.
Mistral AI does not allow their services to be used for the benefit of a third party unless otherwise agreed in a separate contract with them~\cite{mistral_license}.
Meta stipulates that if a third-party product or service exceeds 700 million monthly active users, they must apply for a special license from Meta; otherwise, they are not authorized to use Llama 2~\cite{llama_license}.
Ai2 requires a third party to complete a derivative impact report with all model and data derivatives~\cite{ai2_license}.
Moreover, none of these organizations permit third parties to rebrand their services, such as by obscuring or removing copyright and trademark information.

\section{Theoretical Analysis on Fingerprinting and Model Watermarking}
\label{appendix:theory}

Our \soft utilizes the advantage of both query-based fingerprinting and model watermarking.
In this section, we theoretically analyze and prove that \soft could be more optimal compared with current query-based fingerprinting and model watermarking methods.

\subsection{Query-Based Fingerprinting}
\label{appendix:theory_fingerprint}

Query-based fingerprinting~\cite{GULYO24, JZSLH24} primarily optimizes an LLM-specific adversarial prompt for the candidate model, which can serve as a prefix/suffix to the input query to generate specific outputs.
We first formalize existing state-of-the-art query-based fingerprinting under a unified loss function.
Then we conduct a comprehensive investigation of their inherent limitations.
The notations can be found in \Cref{table:notation}.

\mypara{TRAP}
TRAP~\cite{GULYO24} is based on GCG~\cite{ZWCNKF23}.
Its objective is to optimize an adversarial suffix, denoted as $a$, that manipulates the reference model into producing a predetermined target response $y$, given the input query $x$.
Then, it uses the suffix to query the black-box LLM and compare the output with the target string.
A match indicates that the target candidate model may power the black-box LLM.
TRAP's loss function can be formulated as follows:
\begin{equation}
\begin{aligned}
\label{equation:loss_trap}
\mathcal{L}_{\text{TRAP}}(y, x, a) & = \mathcal{L}(y, h(x\|a)),
\end{aligned}
\end{equation}
where $h \in H$ represents the candidate model's prompt template, and $\|$ denotes the concatenation between texts.
Then, TRAP optimizes the adversarial suffix $a$ to minimize this loss function by iteratively comparing the loss value and replacing the tokens in the suffix.

\mypara{ProFlingo}
ProFlingo~\cite{JZSLH24} improves TRAP in four aspects.
First, ProFlingo uses a different query set.
Second, ProFlingo's optimization objective is an adversarial prefix instead of a suffix.
Third, instead of using only a single prompt template, ProFlingo optimizes several templates simultaneously to improve the algorithm's robustness.
Furthermore, ProFlingo proposes a new way to update the adversarial prefix, which enhances the process of generating adversarial examples and the efficiency of the overall algorithm.
ProFlings's loss function is formulated as follows:
\begin{equation}
\begin{aligned}
\label{equation:loss_proflingo}
\mathcal{L}_{\text{ProFlingo}}(y, x, a) & = \sum_{h\in H}\mathcal{L}(y, h(a\| x)),
\end{aligned}
\end{equation}
where $H$ is a collection of prompt templates that can be manually designed.
The adversarial prefix $a$ is optimized to minimize its loss function.

\mypara{Unified Loss}
Based on \Cref{equation:loss_trap} and \Cref{equation:loss_proflingo}, we present a unified loss function summarizing the above query-based identification methods (\Cref{equation:unified_loss}):
\begin{equation}
\begin{aligned}
\mathcal{L}_{\text{AE}}(y, x, a) = \sum_{h\in H}\mathcal{L}(y, \mathbf{insert}_{z_h}(h(x), a)),
\end{aligned}
\end{equation}
where $H$ is a prompt template set containing different instructional templates with various system prompts, and $h(x)$ represents encoding user input $x$ with the template $h$.
$\mathbf{insert}_z(s_1, s_2)$ denotes inserting sequence $s_2$ into $s_1$ at the position $z_h$.
The position $z_h$ is adaptable for different prompt templates $h$, enabling support for various prompt templates.

\begin{table*}[!t]
\centering
\setlength{\tabcolsep}{1.1mm}
\scalebox{0.80}{
\begin{tabular}{rcl||rcl}
\toprule
\multicolumn{3}{c||}{Variants} & \multicolumn{3}{c}{Functions} \\
\midrule
$m, n, k$ & &{Integers} &$\mathbf{encode}(x)$ & $=$ & $t^x_{1:n}, t^x_i \in V$ \\
$x, y, a$ & &{Texts} & $\mathbf{decode}(t)$ & $=$ & $x^t$ \\
$t_{1:n}$ & &{Sequence of tokens} &$g(t_{1:n})$ & $=$ & $\mathbf{e}_{1:n}, \mathbf{e}_{i} \in \mathbb{R}^d$\\
$\mathbf{e}_{1:n}$ & &{Sequence of token embeddings} &$p_{\Phi}(t_{n+1}|\mathbf{e}_{1:n})$ & $=$ & $g \circ p_{\theta}(t_{n+1}|t_{1:n})$ \\
$V$ & &{Set of all possible tokens} &$p_{\Phi}(t_{n+1:n+m}|\mathbf{e}_{1:n})$ & $=$ & $\prod^{m}_{i=1}p_{\Phi}(t_{n+i}|\mathbf{e}_{1:n+i-1})$ \\
$h$ & &{Prompt template} &$\mathcal{L}^{\text{e}}(t_{n+1:n+m}, \mathbf{e}_{1:n})$ & $=$ & $-\log p_{\Phi}(t_{n+1:n+m} | \mathbf{e}_{1:n})$ \\
$H$ & &{Set of prompt templates} &$p_{\theta}(t_{n+1:n+m}|t_{1:n})$ & $=$ & $\prod^{m}_{i=1}p_{\theta}(t_{n+i}|t_{1:n+i-1})$  \\
$\|$ & &{Text/Sequence concatenation} &$\mathcal{L}(t_{n+1:n+m}, t_{1:n})$ & $=$ & $-\log p_{\theta}(t_{n+1:n+m} | t_{1:n})$\\
$\circ$ & &{Function composition}  \\
$p_{\Phi}(t_{n+1}|\mathbf{e}_{1:n})$ & &{Probability of the next token} \\
$p_{\theta}(t_{n+1}|t_{1:n})$ & &{Probability of the next token} \\
$\mathbf{insert}_z(s_1, s_2)$ & &{Insert $s_2$ into $s_1$ at position $z$} \\
\bottomrule
\end{tabular}
}
\caption{Notations.
}
\label{table:notation}
\end{table*}

\mypara{Limitations of TRAP and ProFlingo}
We compare TRAP and ProFlingo under the unified loss $\mathcal{L}_{\text{AE}}$.
As discussed above, the unified loss can be translated into the $\mathcal{L}_{\text{ProFlingo}}$ loss function when $\mathbf{insert}_{z_h}(h(x), a) := h(a\| x)$.
Likewise, this unified loss can be adapted to $\mathcal{L}_{\text{TRAP}}$ by defining the prompt template collection as $H =\{h\}$, where $h$ represents the default template of the candidate model, and appending $a$ after $x$, i.e., setting $\mathbf{insert}_{z_h}(h(x), a) := h(x\| a)$.
We find $\mathcal{L}_{\text{TRAP}}$ is more restricted as it optimizes to only one prompt template, while $\mathcal{L}_{\text{ProFlingo}}$ optimizes several prompt templates simultaneously (see \Cref{equation:loss_proflingo}), which allows a more generalizable adversarial text that is resilient to larger disturbances.
This could explain why TRAP can only identify the candidate model itself, while ProFlingo can identify more derivatives of the candidate model, as shown in \Cref{section:experiment_compare_fingerprint}.

\noindent Recall that we focus on the task of black-box LLM identification, where interaction with the target model is restricted to text inputs.
The defender optimizes an adversarial prefix/suffix to query the suspect model.
However, because text is composed of discrete tokens, the optimization process inherently operates in a discrete space.
For example, ProFlingo first calculates the gradient for each token $t^a_i$ in the prefix $a$ in each epoch as $\nabla_{t^a_i} \mathcal{L}_{\text{AE}}$.
As we can only replace $t^a_i$ instead of optimizing it directly, ProFlingo updates the adversarial prefix by replacing multiple tokens in each epoch.
Based on the top-$k$ tokens with the most negative gradients, a set of candidate tokens for replacement is sampled.
Subsequently, the loss associated with each potential replacement token is computed.
The token with the minimal loss is then selected.
TRAP follows a similar process, optimizing over a discrete set of inputs.
However, such discrete optimization approaches present significant challenges.
Given the limited number of tokens and the computationally intensive nature of the optimization process, it is difficult to achieve a local optimum by design.

\subsection{Reaching the Optimal of Query-Based Fingerprinting}
\label{appendix:optimal_proof}

Based on the loss function of our \soft and the unified loss \Cref{equation:unified_loss}, we theoretically prove that \soft can reach the optimal of query-based methods by analyzing the optimal loss.

\begin{hypothesis}
\label{hypothesis:injective_appendix}
For most LLMs, we assume their embedding layer $g: V^n \rightarrow \mathbb{R}^{d\times n}$ is injective, where $d$ is the dimension of the token embedding and $n$ is the length of the token sequence.
\end{hypothesis}
In practice, token embeddings in LLMs are often represented as $d$-dimensional vectors.
Consequently, the size of the vocabulary $V$ is much smaller than the embedding space $\mathbb{R}^{d}$.
Such a high dimensionality increases the likelihood that different tokens will have distinct embeddings, thereby making $g$ more likely to be injective.
To empirically verify this hypothesis, we experiment with our three candidate LLMs and the other 18 LLMs.
Specifically, for each LLM, we calculate the cosine similarity between every pair of token embeddings.
We considered two embeddings to be similar if their cosine similarity exceeded 0.99.
The results of our experiments demonstrate that, across all the LLMs tested, the token embeddings were pairwise distinct.
The empirical evidence supports that our \Cref{hypothesis:injective_appendix} holds for most LLMs in practice.

\begin{lemma}
\label{lemma:non_differentiable_appendix}
Based on \Cref{hypothesis:injective_appendix}, the inverse function $q$ of the embedding layer $g$ is non-differentiable.
\end{lemma}
\begin{proof}
Let $g$ be the mapping from tokens to embeddings, and let $q$ be its inverse, defined as:
\begin{equation}
\label{equation:definition_q}
q: \mathbb{R}^{d\times n}\rightarrow V^n,
\end{equation}
where $d$ is the dimension of the token embedding, and $n$ is the length of the embedding sequence.
Note that $V$ is a discrete space, which contains pre-defined tokens that are discontinuous.
Thus, $q$ is non-differentiable due to the discrete nature of its output.
\end{proof}

\begin{lemma}
\label{lemma:equalty_appendix}
Based on \Cref{hypothesis:injective_appendix}, the unified loss $\mathcal{L}_{\text{AE}}$ is derived from \soft's loss $\mathcal{L}^{\text{e}}_{\soft}$ with the application of the inverse function $q$.
\end{lemma}
\begin{proof}
Firstly, we prove $q\circ \mathcal{L}^{\text{e}}(t_{n+1:n+m}, \mathbf{e}_{1:n}) = \mathcal{L}(t_{n+1:n+m}, t^{\mathbf{e}}_{1:n})$.
Based on \Cref{equation:definition_q}, $q$ maps a token embedding sequence to a discrete token sequence.
Thus, $q \circ \mathcal{L}^{\text{e}}(\cdot)$ is a transformation of input domains from $\mathbb{R}^{d\times n}$ to $V^n$.
As $p_{\Phi}(t_{n+1}|\mathbf{e}_{1:n})$ = $g \circ p_{\theta}(t_{n+1}|t_{1:n})$, we can infer $q \circ p_{\Phi}(t_{n+1}|\mathbf{e}_{1:n}) = q \circ g \circ p_{\theta}(t_{n+1}|t_{1:n}) = p_{\theta}(t_{n+1}|t_{1:n})$.
\begin{equation}
\begin{aligned}
q ~ \circ ~ & \mathcal{L}^{\text{e}}(t_{n+1:n+m}, \mathbf{e}_{1:n}) \\
& = q\circ -\log p_{\Phi}(t_{n+1:n+m} | \mathbf{e}_{1:n}) \\
& = q\circ -\sum^{m}_{i=1}\log p_{\Phi}(t_{n+i}|\mathbf{e}_{1:n+i-1}) \\
&= -\sum^{m}_{i=1}q\circ \log p_{\Phi}(t_{n+i}|\mathbf{e}_{1:n+i-1}) \\
& = -\sum^{m}_{i=1}\log p_{\theta}(t_{n+i}|t^{\mathbf{e}}_{1:n+i-1}) \\
&= -\log \prod^{m}_{i=1} p_{\theta}(t_{n+i}|t^{\mathbf{e}}_{1:n+i-1}) \\
& = -\log p_{\theta}(t_{n+1:n+m}|t^{\mathbf{e}}_{1:n}) = \mathcal{L}(t_{n+1:n+m}, t^{\mathbf{e}}_{1:n}).
\end{aligned}
\end{equation}
We conclude $q\circ \mathcal{L}^{\text{e}}(t_{n+1:n+m}, \mathbf{e}_{1:n}) = \mathcal{L}(t_{n+1:n+m}, t^{\mathbf{e}}_{1:n})$.

Then, we prove $q(\mathbf{insert}_z(\mathbf{e}_{1:n}, \mathbf{e}'_{1:m})) = \mathbf{insert}_z(q(\mathbf{e}_{1:n}), q(\mathbf{e}'_{1:m}))$.
Let $z$ be the $k$-th position of $\mathbf{e}_{1:n}$, $k \in [0,n]$, $s_{1:0}$ and $s_{n+1:n}$ denote a null sequence.

\noindent From left:
\begin{equation}
\begin{aligned}
q(\mathbf{insert}_z(\mathbf{e}_{1:n}, \mathbf{e}'_{1:m}))
& = q(\{\mathbf{e}_{1:k}, \mathbf{e}'_{1:m}, \mathbf{e}_{k+1:n}\}) \\
& = t^{\{\mathbf{e}_{1:k}, \mathbf{e}'_{1:m}, \mathbf{e}_{k+1:n}\}} \\
& = \{t^{\mathbf{e}}_{1:k}, t^{\mathbf{e}'}_{1:m}, t^{\mathbf{e}}_{k+1:n}\}.
\end{aligned}
\end{equation}
For the right side:
\begin{equation}
\begin{aligned}
\mathbf{insert}_z(q(\mathbf{e}_{1:n}), q(\mathbf{e}'_{1:m}))
& = \mathbf{insert}_z(t^{\mathbf{e}}_{1:n}, t^{\mathbf{e}'}_{1:m}) \\
& = \{t^{\mathbf{e}}_{1:k}, t^{\mathbf{e}'}_{1:m}, t^{\mathbf{e}}_{k+1:n}\}.
\end{aligned}
\end{equation}
Thus, we conclude $q(\mathbf{insert}_z(\mathbf{e}_{1:n}, \mathbf{e}'_{1:m})) = \mathbf{insert}_z(q(\mathbf{e}_{1:n}), q(\mathbf{e}'_{1:m})).$
Finally, we have:
\begin{equation}
\begin{aligned}
q ~ \circ ~ & \mathcal{L}^{\text{e}}_{\soft}(y, x, \mathbf{e}^a_{1:m}) \\
& = q \circ \sum_{h\in H}\mathcal{L}^{\text{e}}(t^{y}, \mathbf{insert}_{z_h}(\mathbf{e}^{{h(x)}}, \mathbf{e}^a_{1:m})) \\
&= \sum_{h\in H}\mathcal{L}(t^{y}, q(\mathbf{insert}_{z_h}(\mathbf{e}^{{h(x)}}, \mathbf{e}^a_{1:m}))) \\
& = \sum_{h\in H}\mathcal{L}(t^{y}, \mathbf{insert}_{z_h}(q(\mathbf{e}^{{h(x)}}), q(\mathbf{e}^a_{1:m}))) \\
&= \sum_{h\in H}\mathcal{L}(t^{y}, \mathbf{insert}_{z_h}(t^{{h(x)}}, t^a_{1:m})) = \mathcal{L}_{\text{AE}}(y, x, a).
\end{aligned}
\end{equation}
\end{proof}

\begin{lemma}
\label{lemma:loss_closer_to_optimal_appendix}
Given $q$ is non-differentiable (\Cref{lemma:non_differentiable_appendix}), $\mathcal{L}^{\text{e}}_{\soft}$ can reach closer or equal to the optimal solution compared to the unified loss $\mathcal{L}_{\text{AE}}$.
\end{lemma}
\begin{proof}
Based on \Cref{lemma:equalty_appendix}, $\mathcal{L}_{\text{AE}} = q \circ \mathcal{L}^{\text{e}}_{\soft}$.
The non-differentiability of $q$ introduces points in the loss landscape where gradients are either undefined or lead to instability.
This can cause optimization algorithms, particularly those relying on gradient descent, to struggle with convergence, potentially getting stuck in flat regions or diverging near discontinuities.
On the other hand, $\mathcal{L}^{\text{e}}_{\soft}$, provides a smooth loss landscape.
Gradient-based methods can reliably compute gradients and update parameters effectively, allowing for consistent progress toward the optimal solution.
Thus, $\mathcal{L}^{\text{e}}_{\soft}$ is more likely to reach closer or equal to the optimal solution compared to $\mathcal{L}_{\text{AE}}$, which is hindered by the non-differentiable function $q$.
\end{proof}

\begin{claim}
\label{claim:exp_closer_to_optimal_appendix}
Based on \Cref{lemma:loss_closer_to_optimal_appendix}, the optimization of our \soft can be closer to the optimal compared with current query-based methods that adopt the unified loss $\mathcal{L}_{\text{AE}}$.
\end{claim}

\Cref{claim:exp_closer_to_optimal_appendix} shows that our \soft can explore the extent of existing passive identification methods.

\subsection{Model Watermarking}
\label{appendix:theory_watermark}

Model watermarking allows model owners to proactively protect their intellectual property before model distribution.
They mainly fine-tune the LLM to learn a watermark pattern.
Our \soft differs from these methods in the following two aspects.
Firstly, our \soft does not require modifying the model's main transformer layers, thus preserving the model's utility.
This is different from model watermarking, which requires updating the model's parameters to learn the watermark pattern.
The second aspect is that we allow the model owner to define the \tokens.
This not only reduces the collision of \tokens for different models using our \soft, but also provides chances that the corresponding adversarial token embeddings could remain after further fine-tuning.

\section{Challenges and Our Vision}
\label{appendix:challenges}

\begin{table*}[!t]
\centering
\setlength{\tabcolsep}{1.1mm}
\scalebox{0.80}{
\begin{tabular}{ll}
\toprule
Model Name & Model Path \\
\midrule
\textbf{Llama}~\cite{TLIMLLRGHARJGL23} & yahma/llama-7b-hf \\
Vicuna~\cite{Vicuna} & lmsys/vicuna-7b-v1.3 \\
Guanaco~\cite{DPHZ23}& TheBloke/guanaco-7B-HF\\
\midrule
\textbf{Llama2}~\cite{TMSAABBBBBBBCCCEFFFFGGGHHHIKKKKKKLLLLLMMMMMNPRRSSSSSTTTWKXYZZFKNRSES23} & meta-llama/Llama-2-7b-hf\\
Llama2-Chat~\cite{TMSAABBBBBBBCCCEFFFFGGGHHHIKKKKKKLLLLLMMMMMNPRRSSSSSTTTWKXYZZFKNRSES23} & meta-llama/Llama-2-7b-chat-hf \\
Vicuna2~\cite{Vicuna} & lmsys/vicuna-7b-v1.5\\
CodeLlama~\cite{RGGSGTALRRKEBBCGXDCATMUSS23} & codellama/CodeLlama-7b-hf\\
CodeLlama-Python~\cite{RGGSGTALRRKEBBCGXDCATMUSS23} & codellama/CodeLlama-7b-Python-hf \\
CodeLlama-Instruct~\cite{RGGSGTALRRKEBBCGXDCATMUSS23} & codellama/CodeLlama-7b-Instruct-hf\\
Tulu2~\cite{IWPLPDJWSBH23} & allenai/tulu-2-7b \\
Tulu2-DPO~\cite{IWPLPDJWSBH23} & allenai/tulu-2-dpo-7b\\
SciTulu~\cite{IWPLPDJWSBH23} & allenai/scitulu-7b \\
Orca2~\cite{MCMCSACRJAPZRKA23} & microsoft/Orca-2-7b\\
Llama2-Chat-Un~\cite{Uncensored_Models} & georgesung/llama2\_7b\_chat\_uncensored \\
\midrule
\textbf{Mistral}~\cite{JSMBCCBLLSLLSSLWLS23} & mistralai/Mistral-7B-v0.1 \\
Mistral-Instruct~\cite{JSMBCCBLLSLLSSLWLS23} & mistralai/Mistral-7B-Instruct-v0.1 \\
Zephyr~\cite{TBLRRBHWFHSSRW23} & HuggingFaceH4/zephyr-7b-beta\\
Mistral-v0.3~\cite{JSMBCCBLLSLLSSLWLS23} & mistralai/Mistral-7B-v0.3\\
Mistral-v0.3-Instruct~\cite{JSMBCCBLLSLLSSLWLS23} & mistralai/Mistral-7B-Instruct-v0.3 \\
\midrule
\textbf{Gemma}~\cite{T24} & google/gemma-7b \\
Gemma-Instruct~\cite{T24} & google/gemma-7b-it \\
Zephyr-Gemma~\cite{TBLRRBHWFHSSRW23} & HuggingFaceH4/zephyr-7b-gemma-v0.1\\
Zephyr-Gemma-SFT~\cite{TBLRRBHWFHSSRW23} & HuggingFaceH4/zephyr-7b-gemma-sft-v0.1\\
\midrule
\textbf{OpenLlama}~\cite{OpenLLaMA}  & openlm-research/open\_llama\_7b \\
OpenLlama2~\cite{OpenLLaMA} & openlm-research/open\_llama\_7b\_v2\\
\midrule
\textbf{Llama3}~\cite{T24} & meta-llama/Meta-Llama-3-8B \\
Llama3-Instruct~\cite{T24} & meta-llama/Meta-Llama-3-8B-Instruct \\
\midrule
\textbf{OLMo}~\cite{GBWBKTJIMWAAACCDEGHKMMMNNPPRSSSSSWDLRZDLSSH24} & allenai/OLMo-7B-hf\\
OLMo-SFT~\cite{GBWBKTJIMWAAACCDEGHKMMMNNPPRSSSSSWDLRZDLSSH24}& allenai/OLMo-7B-SFT-hf\\
OLMo-Instruct~\cite{GBWBKTJIMWAAACCDEGHKMMMNNPPRSSSSSWDLRZDLSSH24}& allenai/OLMo-7B-Instruct-hf \\
\midrule
\textbf{GPT4-Turbo}~\cite{O23} & \gpt \\
\midrule
\textbf{Claude3.5-Sonnet}~\cite{claude3.5} & \claude \\
\bottomrule
\end{tabular}
}
\caption{Details of the LLMs used in the experiments.
The LLMs in one group indicate that they are derivatives of the LLMs in bold.
Note that the OpenLlama series uses the same framework as Llama, but they are pre-trained on different datasets.
}
\label{table:model_names}
\end{table*}

Our comprehensive experiments on various methods have illustrated the difficulties of identifying the origin of a black-box LLM.
Drawing from theoretical and empirical analysis in this study, we highlight the challenges and opportunities.

\mypara{Challenges}
The first challenge is that the auditor may not have sufficient knowledge of practical customizations applied to LLMs.
For example, the LLM may be further fine-tuned with specific system prompts and equipped with specific hyperparameters~\cite{GULYO24, JZSLH24}.
The model's structure and weights can also be modified or rearranged~\cite{ZZWL23}.
For example, existing methods and \soft are proven ineffective for models that adopt the Llama architecture yet are trained independently (e.g., OpenLlama).

The second challenge is the black-box access to the suspect model.
It implies that we can only query the model with text and obtain a text response.
This is inevitable, as the malicious third party tends to conceal the model's true origin.
Such black-box access inevitably limits the identification methods, as they are constrained to querying the model and analyzing the text outputs.
For example, existing query-based fingerprinting relies on discrete optimization, as they cannot directly query the LLM with embeddings but with discrete texts.
As analyzed in \Cref{appendix:theory}, such discrete optimization makes it hard to reach the optimum, leading to performance degradation.
Additionally, the adapter variant of IF~\cite{XWMKXC24} is not applicable as we are unable to insert an F-Adapter into the black-box LLM.
Our experimental results show that \soft surpasses existing methods, though its efficacy declines when substantial weight modifications are present.

\mypara{Our Vision}
As we delve deeper into methods for identifying black-box LLMs, there is still much to be done to safeguard the intellectual property rights of developers and owners of LLMs, as well as other foundation models~\cite{LLWL23, CLDZHWZL23, WLYHQWJYZSXXLDDT23, YYZWCZCLZHCZZZHZZCHZLLS24}.
To begin with, we present a real example of two large vision-language models (VLMs) as an introduction.
In June 2024, a user on GitHub raised concerns that the Llama3-V model may have improperly copied elements from the MiniCPM-Llama3-V 2.5 project~\cite{YYZWCZCLZHCZZZHZZCHZLLS24}, prompting a thorough investigation by the authors of the accused model.\footnote{\url{https://github.com/OpenBMB/MiniCPM-V/issues/196}.}
The inquiry revealed striking similarities in code structure and minor variable name changes between Llama3-V and MiniCPM-Llama3-V 2.5.
Further similarities were found in behavior, particularly in specialized tasks such as recognizing Tsinghua Bamboo Characters, which were from the in-house training data of MiniCPM-Llama3-V 2.5.
This serves as a clear indicator of the infringement and exemplifies the potential benefits of proactive identification strategies.

Beyond technological solutions, we emphasize the importance of nurturing a robust open-source community.
Such communities not only serve as guardians, preventing the unauthorized use of copyrighted material, but also as incubators for technological innovation.
The incident with Llama3-V highlights the pivotal role that the community can play in uncovering and addressing infringement.
Moreover, by participating in open-source projects, researchers and developers gain access to a collaborative environment where they can innovate and refine technologies at a pace that individual efforts cannot match.

\begin{table*}[!t]
\centering
\setlength{\tabcolsep}{1.1mm}
\scalebox{0.80}{
\begin{tabular}{llccccc}
\toprule
Dataset & Model & Original &\soft & IF & \soft-CSN & IF-CSN \\
\midrule
\multirow{3}{*}{HellaSwag}&Llama & 0.57±0.005 & 0.57±0.005 & 0.57±0.005 & 0.57±0.005& 0.57±0.005\\
&Llama2 & 0.57±0.005 & 0.57±0.005 & 0.58±0.005 & 0.58±0.005& 0.58±0.005\\
&Mistral & 0.61±0.005 & 0.61±0.005 & 0.61±0.005 & 0.60±0.005& 0.60±0.005\\
\midrule
\multirow{3}{*}{MMLU}&Llama & 0.32±0.059 & 0.32±0.059 & 0.31±0.057 & 0.32±0.059& 0.33±0.059 \\
&Llama2 & 0.41±0.088 & 0.41±0.088 & 0.41±0.089 & 0.40±0.085& 0.42±0.085 \\
&Mistral & 0.60±0.133 & 0.60±0.133 & 0.60±0.135 & 0.57±0.126& 0.59±0.115 \\
\bottomrule
\end{tabular}
}
\caption{Utility on HellaSwag and MMLU.
[Method]-CSN denotes that the model is first proactively modified by [Method] and then fine-tuned on CSN.}
\label{table:utility_watermark}
\end{table*}

Furthermore, while current licensing practices for foundation models are extensive (see \Cref{table:licenses}), there is a notable deficiency in the regulation of the training data used to develop these models.
The vast corpora required to train such models are often neither open-sourced nor clearly documented~\cite{introducing_gpts, TLIMLLRGHARJGL23, JSMBCCBLLSLLSSLWLS23}, raising significant concerns about copyright compliance and transparency.
This lack of oversight and regulation can lead to ambiguities in ownership and ethical use, which are crucial for maintaining trust and integrity within the AI community.
To address these challenges, we advocate for clearer legislation and more stringent regulatory frameworks that not only keep pace with technological advancements but also ensure that the copyrights of large training datasets are adequately protected~\cite{eu_ai_act_r58}.

\section{Experimental Settings}
\label{appendix:settings}

\Cref{table:model_names} shows the details of the models used in our experiments, including the versions we used.

As for the dataset, CodeSearchNet (CSN)~\cite{HWGAB19} is a widely used dataset that contains about six million functions from open-source code spanning six programming languages (Go, Java, JavaScript, PHP, Python, and Ruby).

\section{Experiments}
\label{appendix:experiments}

\subsection{Utility on Original, Modified, and Fine-tuned LLMs}
\label{appendix:utility}

We evaluate the performance of all original, modified, and fine-tuned LLMs on HellaSwag and MMLU datasets, which are widely used as benchmarks for evaluating LLMs.
\Cref{table:utility_watermark} shows the results.
We observe that IF does not significantly affect the model's utility, while our \soft keeps exactly the same results on both datasets under the same random seed.
This is because our \soft does not modify the model's main transformer layers, so there will be no difference if there are no \tokens in the input.
The results also support that our fine-tuning does not affect the model's utility.

\subsection{AUC/ROC Results of Three LLMs}
\label{appendix:roc}

\Cref{figure:roc_fingerprint} shows the AUC/RUC Curve of each method on three different LLMs.
As shown in the results, our \soft achieves an AUC of 1.00 on both Llama and Mistral, outperforming the other two methods.
On Llama2, \soft achieves an AUC of 0.91, while the AUCs of TRAP and ProFlingo are 0.55 and 0.93.
This further shows the effectiveness of our method.

\subsection{TSNE Visualization of Three LLMs}
\label{appendix:tsne}

\Cref{figure:tsne_embeddings} shows the visualization results of the adversarial embeddings and other randomly selected token embeddings of all three LLMs.
We observe that the adversarial embedding for each LLM is within the circle of its token embeddings; there are even several original token embeddings behaving like outliers.
In this sense, it is challenging to detect the adversarial embedding using outlier detection methods, exhibiting the stealthiness of our \soft.

\subsection{Counterexample Analysis}
\label{appendix:counterexample}

Although our method outperforms most baselines, there are still some counterexamples.
Based on Llama2, the TRRs of CodeLlama and CodeLlama-Python are 0.00, while the TRR of CodeLlama-Instruct is 0.06.
Thus, we further examined the reason why \soft performs worse on these models.
Specifically, we measure the difference between the candidate model and its derivatives.
We compute the cosine distance between the weights of the two models.
\Cref{figure:performance_with_similarity} shows the relationship between the TRRs and the cosine distance between the related suspect models and the candidate model.
A larger cosine distance indicates a bigger weight shift from the candidate model.
We observe that the performance of \soft would drop to around 0.00 when the cosine distance is too large.
For example, the cosine distance between CodeLlama and Llama2 reaches 0.3653, which is much higher than other Llama2 derivatives, e.g., Tulu2 (0.0009).
\soft only achieves a TRR of 0.00 on CodeLlama while reaching a TRR of 0.64 on Tulu2.
This suggests that \soft could be dodged when there is a significant weight shift from the candidate model.

Additionally, OpenLlama presents a permissively licensed open-source reproduction of Meta AI's Llama model, i.e., it uses Llama's structure and is trained on another dataset by another organization.
We compute the cosine distance between it and its derivatives with Llama as well, and label them as Llama (structure).
Results in \Cref{figure:performance_with_similarity} show that their weights are more distant from Llama, and \soft failed to identify them as Llama.
This provides a new scenario where the adversary steals the architecture of an LLM illicitly, trains it from scratch, and distributes it for unauthorized use, e.g., commercial purposes.
We do not include a discussion on this topic and leave it as future work.

\begin{figure}[!t]
\centering
\centering
\includegraphics[width=0.75\columnwidth]{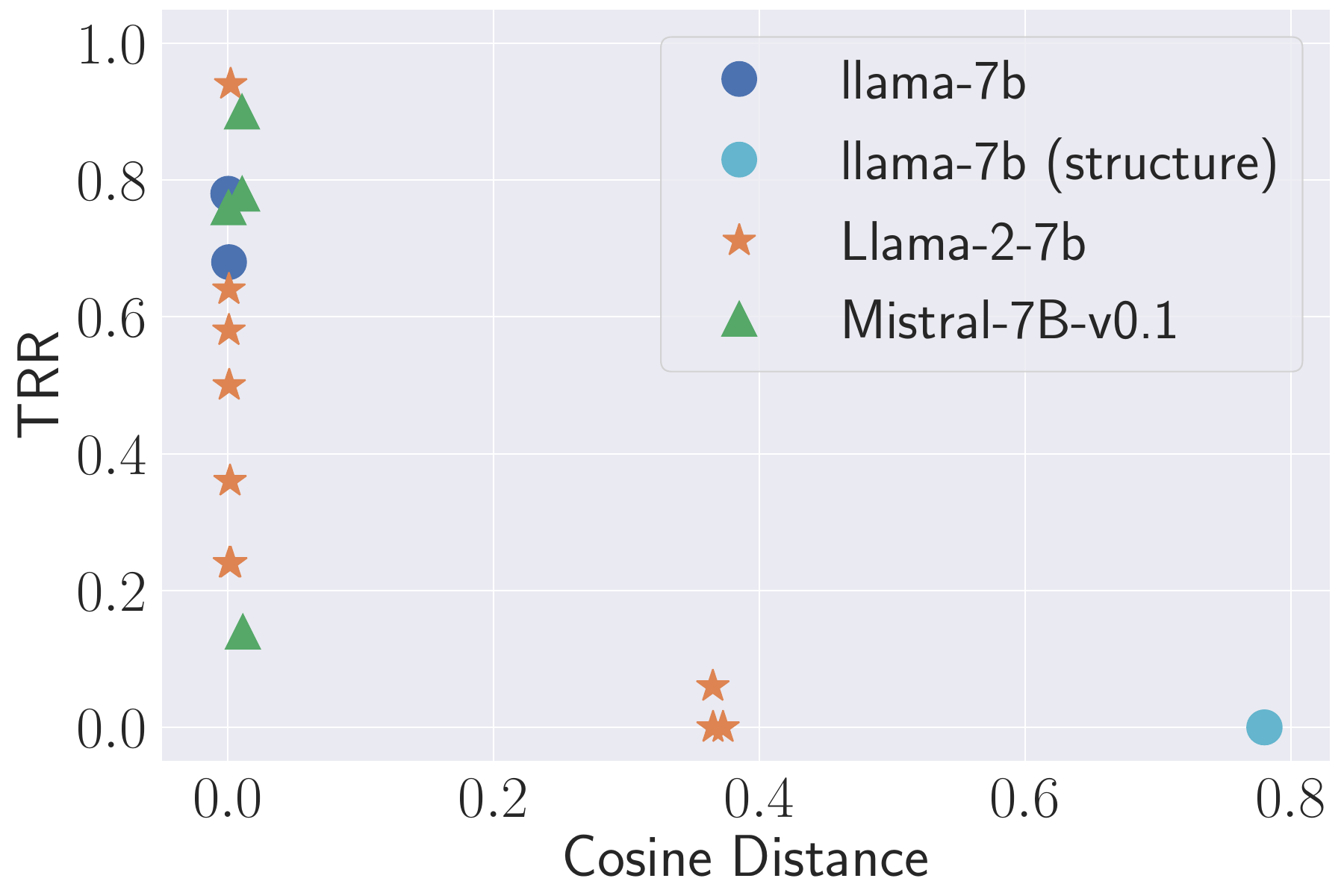}
\caption{Relationship of TRRs and the cosine distances between the suspect models and the candidate model.
}
\label{figure:performance_with_similarity}
\end{figure}

\begin{figure*}[!t]
\centering
\begin{subfigure}{0.56\columnwidth}
\includegraphics[width=\columnwidth]{figures/roc_llama_7b.pdf}
\centering
\caption{Llama}
\label{figure:roc_llama_7b}
\end{subfigure}
\begin{subfigure}{0.56\columnwidth}
\includegraphics[width=\columnwidth]{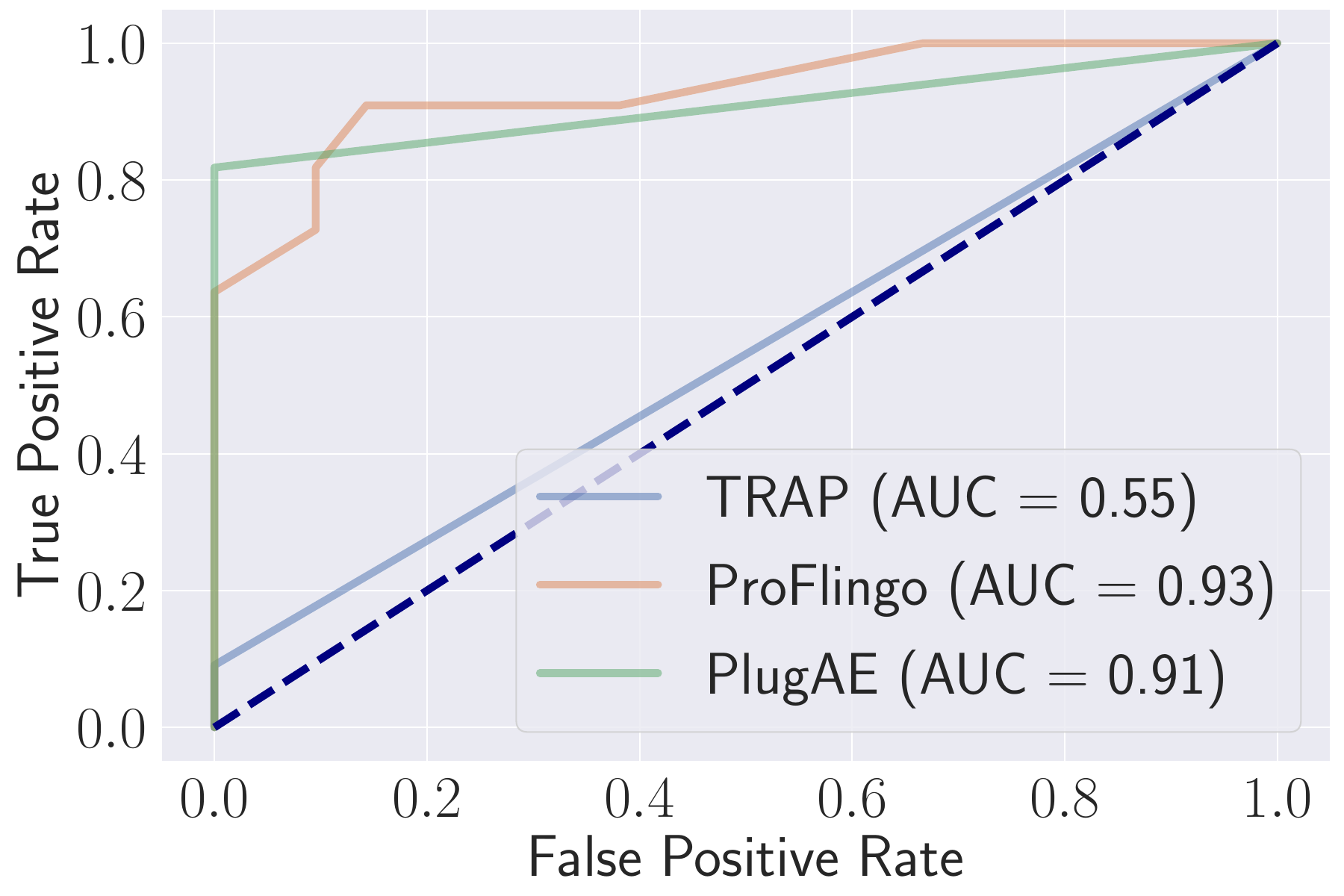}
\centering
\caption{Llama2}
\label{figure:roc_llama2_7b}
\end{subfigure}
\begin{subfigure}{0.56\columnwidth}
\includegraphics[width=\columnwidth]{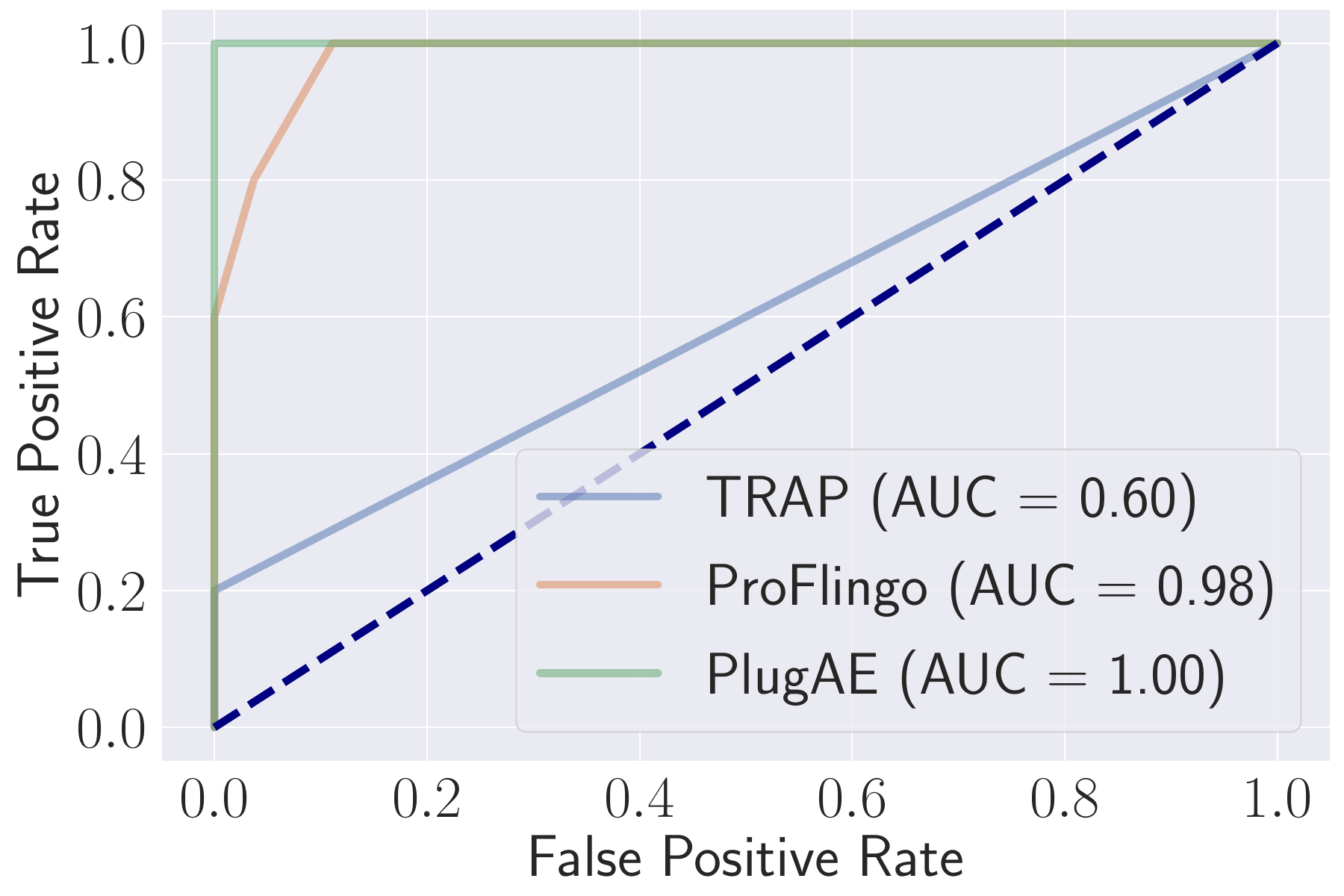}
\centering
\caption{Mistral}
\label{figure:roc_mistral_7b}
\end{subfigure}
\caption{ROC curve of different methods on three candidate LLMs.}
\label{figure:roc_fingerprint}
\end{figure*}

\begin{figure*}[!t]
\centering
\begin{subfigure}{0.56\columnwidth}
\includegraphics[width=\columnwidth]{figures/tsne_llama_7b.pdf}
\centering
\caption{Llama}
\label{figure:tsne_llama_7b}
\end{subfigure}
\begin{subfigure}{0.56\columnwidth}
\includegraphics[width=\columnwidth]{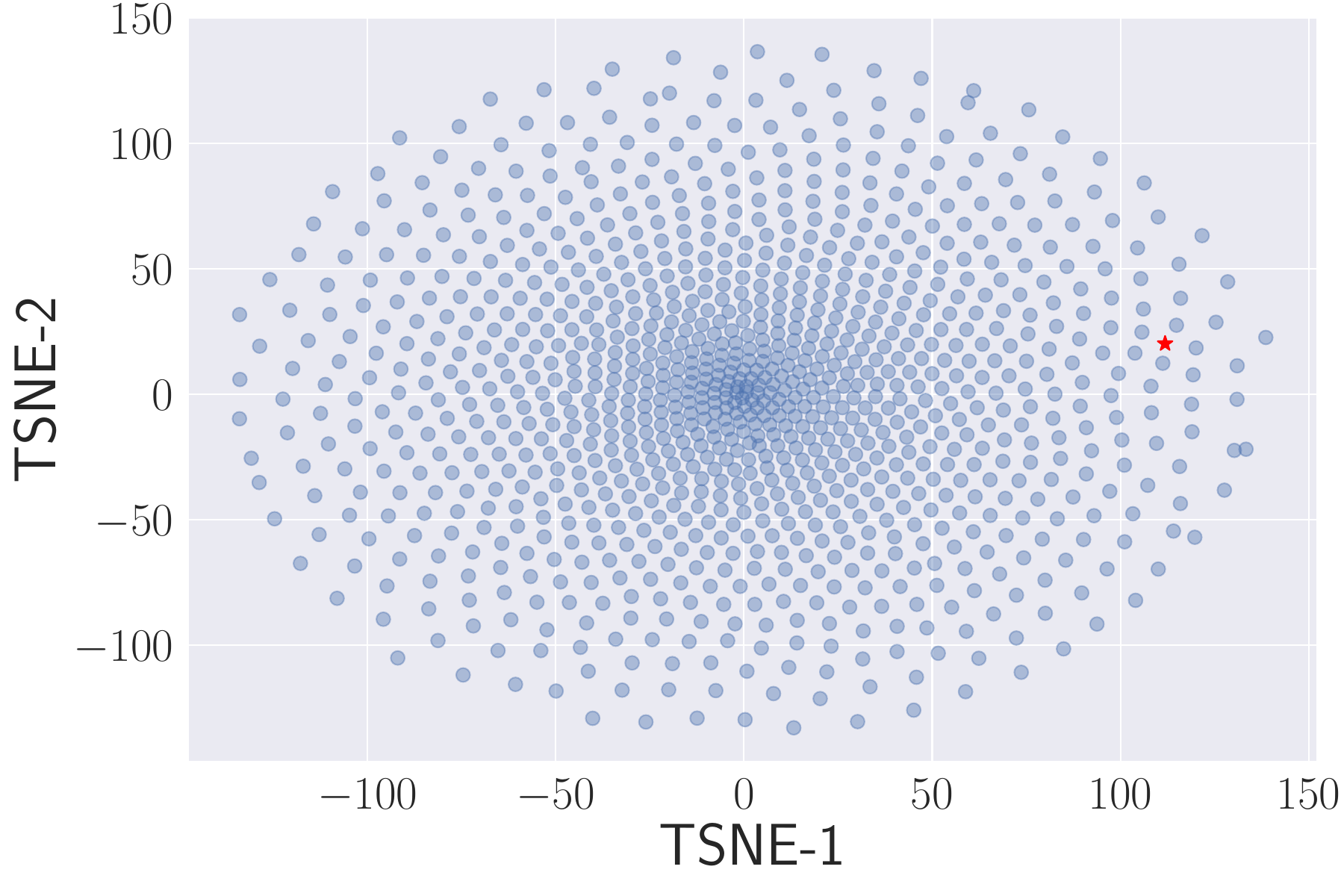}
\centering
\caption{Llama2}
\label{figure:tsne_llama2_7b}
\end{subfigure}
\begin{subfigure}{0.56\columnwidth}
\includegraphics[width=\columnwidth]{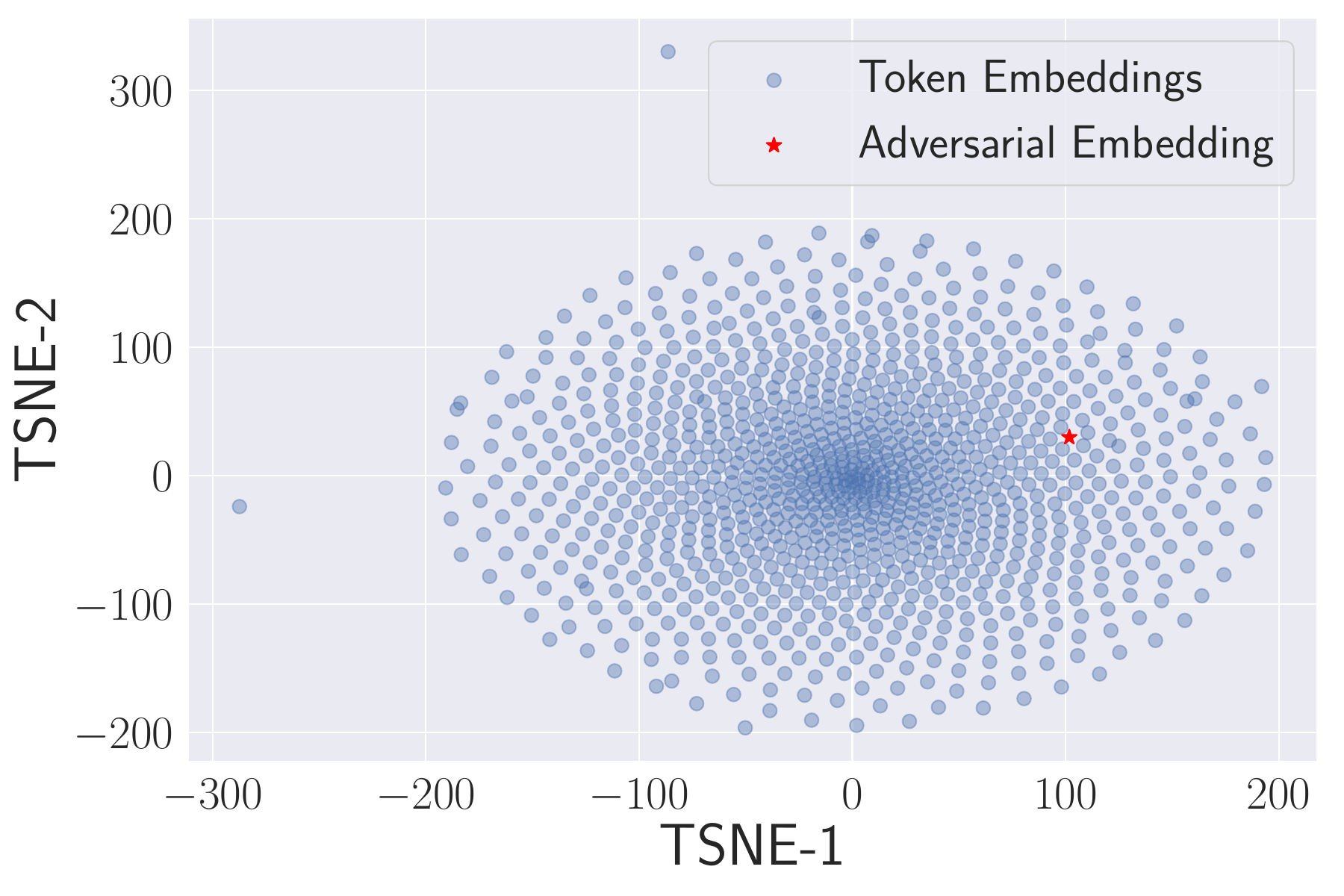}
\centering
\caption{Mistral}
\label{figure:tsne_mistral_7b}
\end{subfigure}
\caption{T-SNE visualization of the adversarial embedding and 1,000 randomly sampled token embeddings in different candidate LLMs.}
\label{figure:tsne_embeddings}
\end{figure*}

\begin{table*}[!t]
\centering
\setlength{\tabcolsep}{1.1mm}
\scalebox{0.80}{
\begin{tabular}{lp{12cm}}
\toprule
Style Name & System Prompt \\
\midrule
Empty & \na  \\
Basic & ``You are a concise AI assistant. Give short, direct answers with minimal explanation unless the user asks for more detail.'' \\
Helpful & ``You are a helpful, respectful, and honest assistant. Your goal is to help the user solve problems efficiently. Be polite, practical, and focus on giving useful step-by-step guidance.'' \\
Alpaca & ``Below is an instruction from a user. Provide a detailed and thoughtful response that completes the request.'' \\
ChatGPT & ``You are ChatGPT, a conversational AI created by OpenAI. You give clear, structured answers, explain reasoning when helpful, and maintain a friendly but professional tone.'' \\
\bottomrule
\end{tabular}
}
\caption{
Different system prompt styles.
}
\label{table:system_prompts}
\end{table*}

\end{document}